\documentclass[12pt]{article}
\usepackage[margin=1.25in]{geometry}
\usepackage{amsmath,amsthm,amssymb}
\usepackage{times}
\usepackage{enumerate}
\usepackage{setspace}
\usepackage[colorlinks,linkcolor=blue,anchorcolor=blue,citecolor=blue]{hyperref}
\usepackage{comment}
\usepackage{dsfont}
\usepackage{float}
\usepackage[title]{appendix}
\usepackage{natbib}
\usepackage{graphicx}
 
\usepackage{tikz}
\usetikzlibrary{arrows.meta}
\usetikzlibrary{positioning, fit, calc, patterns}   
\tikzset{block/.style={draw, thick, minimum width=3cm ,minimum height=1.3cm, align=center},   
	line/.style={-latex}     
}  
 
\theoremstyle{plain}

\newtheorem{proposition}{Proposition}
\newtheorem{theorem}{Theorem}
\newtheorem{corollary}{Corollary}

\newtheorem{lemma}{Lemma}

\newtheorem{definition}{Definition}
\newtheorem{example}{Example}
\newtheorem{remark}{Remark}

\newcommand{\R}{\mathbb{R}}

\makeatletter

\newcommand{\Rmnum}[1]{\expandafter\@slowromancap\romannumeral #1@}

\title{Secret Communication with Plausible Deniability\footnote{We are grateful to Xiao Lin and Frank Yang for inspiring conversations at an early stage of this project. We thank Arjada Bardhi, Yi Chen, Jeff Ely, R.Vijay Krishna, Fei Li, Jo\~{a}o Thereze, Mark Whitmeyer, and participants at UNC theory workshop for insightful comments.}}
\author{Xiaoyu Cheng\footnote{Department of Economics, Florida State University, Tallahassee, FL, USA. E-mail: \href{mailto:xcheng@fsu.edu}{xcheng@fsu.edu}.}\and Yonggyun Kim\footnote{Department of Economics, Florida State University, Tallahassee, FL, USA. E-mail: \href{mailto:ykim22@fsu.edu}{ykim22@fsu.edu}.}
\and Michael P.H. Tam\footnote{Department of Economics, University of North Carolina, Chapel Hill, NC, USA. E-mail: \href{mailto:pakhang@unc.edu}{pakhang@unc.edu}}
}

\begin{document}
\maketitle

\begin{abstract}
   Communication is secret if a message is independent of the state; however, the receiver's subsequent action may still reveal that she has acted on hidden information. This paper studies when secret communication can also provide plausible deniability: under single-crossing preferences, every action induced by the sender's message must be rationalizable using the receiver's baseline information alone. We characterize joint information structures that satisfy both secrecy and plausible deniability. We show that plausible deniability restricts communication exactly when the baseline message is directional---meaning its likelihood is monotone in the state. Combining this restriction with secrecy, we show that, for directional messages, frontier communication reveals at most whether the state lies above or below a cutoff. Finally, we identify conditions under which a greatest feasible communication structure exists and can be constructed explicitly in a simple way.
\end{abstract}

\newpage 
\section{Introduction}

Communication often takes place under conditions of surveillance. A sender may wish to convey information to a receiver while ensuring that an outside observer learns nothing about the underlying state, either because the information is sensitive or because the communication must, for legal, institutional, or optimal reasons, be uninformative about the state. This concern is captured by the notion of \emph{(perfect) secrecy}, long studied in cryptography and information theory \citep{shannon1949}. In the language of information structures, a message is secret if its distribution is statistically independent of the state: an outside observer learns nothing from the message, yet it may still transmit useful information to the receiver through its correlation with the receiver's own information.

Secrecy, however, provides only partial protection. Even when the message is statistically uninformative to outsiders, the receiver's subsequent action may reveal that she has acted on information privately conveyed to her. Especially, information often matters precisely by changing behavior: an investor trades, a regulator investigates, a manager intervenes, or a protector takes precautionary action. Such responses may be informative even when the message that triggered them is not. The Martha Stewart/ImClone case illustrates this tension: Stewart's sale of all her ImClone shares shortly before adverse news about the company became public drew scrutiny because it made the possible receipt of a private tip salient to investigators \citep{sec2003stewart}. Therefore, the challenge is that while a message is secret, the action it induces can be revealing.

This paper studies how secret communication can be designed to address this challenge. We introduce the notion of \emph{plausible deniability}: whenever the receiver acts on the sender's message, her action must still be justifiable as a rational response to information she already possessed, without invoking that message.\footnote{The Stewart/ImClone case also illustrates this justification idea. According
to the SEC, Stewart claimed that she sold her ImClone shares because of a prior agreement to sell if the stock fell below \$60, an explanation takes exactly the form of plausible deniability considered here: it attempts to justify the observed action as a rational response to a publicly observable baseline message---the stock price reaching the \$60 threshold---without invoking the alleged hidden tip. According to the SEC, the problem was not that such a baseline-based justification would be irrelevant, but that the
asserted prior agreement was fabricated.} Plausible deniability therefore protects not the message itself, but the receiver's ability to act on useful information without exposing its source. This protective role is central in reporting and compliance settings, where an intervention may need to respond to sensitive information from an employee or whistleblower while shielding that person from retaliation or unwanted scrutiny. Together, secrecy and plausible deniability provide two complementary layers of protection: the message is statistically uninformative on its face, and the receiver's response remains explainable without it. We characterize the maximally informative communication structures that satisfy both requirements.

Formally, we study a communication problem with a sender, a receiver, and an outside observer whose presence constrains communication. The receiver has access to a baseline information structure, whose realized message is observable to all parties. The sender communicates by designing the joint distribution of his message and the receiver's baseline message, while preserving the publicly known marginal distribution of each. Only the sender and the receiver know the dependence between the two messages. Secrecy requires the sender's message to have a state-independent marginal distribution. Hence, any information conveyed to the receiver must come entirely from the sender's message through its correlation with the receiver's baseline information.

After observing both messages, the receiver chooses an optimal action given her posterior belief. Her action is publicly observable, but her preferences are private information. Nevertheless, it is commonly known that her utility function satisfies single crossing, so that she prefers higher actions in higher states, and that she takes a given default action under the prior belief. For example, an investor may be known to take more aggressive buy positions in higher states, while the precise beliefs at which she switches from one position to another depend on her privately known risk preferences. In formal terms, plausible deniability restricts how the sender's message can be correlated with the receiver's baseline message:  any action induced by observing both messages must also be rationalizable from the baseline message alone, for some utility function consistent with the publicly known restrictions.

The sender's goal is to choose, subject to secrecy and plausible deniability, a joint information structure that is maximally informative to the receiver, in the sense that it gives her a higher expected payoff across all decision problems and all preferences consistent with public knowledge. Our main results characterize the frontier of feasible joint information structures and provide conditions under which a greatest feasible structure exists and can be constructed explicitly in a simple way.

More specifically, we show that plausible deniability imposes a restriction on the joint information structure if and only if the baseline message is \emph{directional}, in the sense that its likelihood is monotone in the state.\footnote{When the baseline message is non-directional, that is, non-monotone in the state, plausible deniability imposes no restriction on the state-dependence of the joint message; see Section \ref{sec:PD}.} In that case, plausible deniability requires the likelihood of the joint message to be monotone in the same direction (Theorem \ref{thm:PD-characterization}). Thus, plausibly deniable communication can nudge the receiver only in the direction made defensible by her baseline information: it may strengthen or soften the response suggested by the baseline message, but it cannot justify an action on the opposite side of the default action.

Our analysis of secrecy uses \cite{green2022two}'s signal representation of the baseline information structure, which turns the secrecy constraint into a rearrangement problem over signal realizations. Building on this representation, we show that signal-based structures are without loss for the sender's frontier problem under secrecy and plausible deniability: every frontier point is payoff-equivalent to some signal-based structure (Theorem
\ref{thm:frontier_secrecy_PD} and Corollary \ref{cor:frontier_secrecy_PD}). Importantly, this reduction yields a simple posterior geometry: directional baseline messages induce posteriors that are lower- or upper-tail truncations of the prior, while non-directional messages induce posteriors given by the prior conditional on subsets of states. Thus, for a directional message, frontier communication takes an extreme form: it reveals at most whether the state lies above or below a cutoff.

This posterior geometry also highlights a dual role of plausible deniability. It restricts feasible communication, but it also makes the frontier tractable. Specifically, it simplifies an otherwise complex analysis by imposing monotonicity on the likelihoods induced by directional baseline messages (see Section \ref{sec:secrecy_signal}). This observation motivates our study of almost-directional baseline information structures, in which non-directionality is limited to at most one message. This class includes all information structures with at most three messages that satisfy the monotone likelihood ratio property.

For almost-directional baseline structures, we show that a greatest feasible communication exists and can be constructed explicitly from a \cite{green2022two} signal representation by ordering messages according to their direction: decreasing, non-monotone, and increasing (Theorem \ref{thm:directional_almost_directional}). 

The same direction-ordered construction is robust beyond the almost-directional benchmark: Theorem \ref{thm:PD_greatest_if} gives a sufficient condition under which the non-monotone messages are sparse enough for this construction to remain greatest. When a greatest feasible communication is obtained by these results, it has a sharp posterior geometry: every induced posterior is the prior conditional on an interval of states.

Beyond the sender's communication problem, the same logic speaks to a broader question of protective information design: how can useful information be elicited and acted upon without exposing its source? The model can be read from this perspective as describing receiver-designed reporting protocols. A compliance officer, for example, may have baseline information about workplace misconduct and may design a routine reporting channel for employees or whistleblowers. The report is uninformative on its own, but its dependence with the officer's baseline information can still guide the officer's subsequent intervention. Plausible deniability requires that this intervention remain justifiable from the officer's baseline information alone. Our results describe how informative such protective protocols can be: when the officer's baseline information is directional, the most informative feasible refinements take a cutoff form; under broader conditions, the direction-ordered construction yields a greatest feasible protocol explicitly.

\medskip 

The remainder of the paper is organized as follows. Section \ref{sec:related_literature} discusses related literature. Section \ref{sec:model} introduces the model and defines secrecy and plausible deniability. Section \ref{sec:PD-secrecy} characterizes the implications of these requirements on the joint information structure. Section \ref{sec:frontier_optimal} identifies the frontier of joint information structures satisfying both requirements, and provides conditions for the existence and construction of a greatest feasible communication. Section \ref{sec:discussions} provides further discussions, and Section \ref{sec:conclusion} concludes. All proofs are in the appendix. 

\subsection{Related Literature}\label{sec:related_literature}

Our paper contributes to a growing literature that characterizes maximally informative information structures under additional constraints. Closest to us, \cite{strack2024} and \cite{he2026private} characterize Blackwell or Blackwell-Pareto frontiers under privacy constraints, requiring signals to be statistically uninformative about specified variables.\footnote{See also \cite{xu2025aprivacy,xu2025bprivacy} for recent characterizations of Blackwell frontiers under privacy constraints, and \cite{wang2024inferentially} for Blackwell-optimal information release under inferential privacy and private-private information constraints.} Our secrecy constraint, following the notion of perfect secrecy in \cite{shannon1949}, has a similar independence form, but our informativeness criterion is fundamentally different; see Section \ref{sec:relation_to_privacy} for a more detailed comparison. 
Also closely related is \cite{cavounidis2025}, who study when perfect secrecy permits full revelation of the state; see Section \ref{sec:secrecy_signal} for a related discussion.\footnote{Applications of individually uninformative but jointly revealing information structures to mechanism design include \cite{zhu2023private} and \cite{larionov2025first}.} Our focus is different: plausible deniability may rule out full revelation even when secrecy alone permits it, and our frontier characterization shows how this additional requirement shapes the most informative feasible communications.

Our notion of plausible deniability is conceptually related to work in which communication is constrained by what an outside observer can infer. Most closely, \cite{antic2025a} study communication under scrutiny, where agents exchange information in the presence of an observer with partially opposed interests and must preserve plausible deniability. Relatedly, \cite{chassang2019crime} show how plausible deniability can facilitate information transmission when a monitor faces retaliation or reputational concerns. We study a different strategic environment and impose plausible deniability directly as an information-design constraint.

Finally, our analysis is related to work on how multiple information sources interact. \cite{borgers2013signals} define when information sources are complements or substitutes, emphasizing that the value of one source depends on how it interacts with another. \cite{cheng2024diversity} study joint information structures with fixed marginals and show how dependence affects informativeness. Relatedly, \cite{de2023robust} study robust decision making when marginal information structures are known but their correlation is not. In our setting, the receiver's baseline information structure is fixed, while secrecy requires the sender's marginal message to be uninformative about the state; the sender's design problem is therefore to choose the dependence between the two messages, subject to plausible deniability.

\section{The Model}\label{sec:model}

We study a communication problem involving an information sender (he), a privately informed receiver (she), and an outside observer (they). The sender aims to convey information to the receiver while maintaining \textit{secrecy} and \textit{plausible deniability} with respect to the outside observer, which are the main constraints of our model and will be formally defined in Section \ref{subsec:S-PD}. 

\subsection{Setup}

The state space is a finite set $\Omega = \{\omega_1, \omega_2, \ldots, \omega_n\}$. The receiver's action space is $A = \{a_{-l}, \ldots, a_{0}, \ldots, a_L\}$ with $l, L \ge 1$. Both states and actions are ordered by their indices. The common prior belief over the states is $\mu \in \Delta(\Omega)$, which has full support. Independently of the sender's communication, the receiver has access to a \emph{baseline} information structure $f: \Omega \to \Delta(X)$ that generates a baseline message $x \in X$. Both $f$ and its realization $x$ are publicly observed.

\paragraph{Sender's communication and covert channel}

The sender communicates by choosing an information structure $g: \Omega \rightarrow \Delta (Y)$ over an arbitrary finite message space $Y$. This communication channel is public in the sense that both the choice of $g$ and its realization $y$ are publicly observed. 

However, the sender can covertly transmit \emph{additional} information by choosing how $x$ and $y$ are \emph{jointly} generated, subject to their publicly observed marginals. Specifically, the sender may select any joint information structure $h: \Omega \rightarrow \Delta (X \times Y)$ that is consistent with $f$ and $g$, i.e., for every $\omega \in \Omega$, 
\begin{equation*}
    f(x|\omega) = \sum_{y \in Y} h(x,y|\omega), \quad g(y|\omega) = \sum_{x \in X} h(x,y|\omega), \quad \forall\,x\in X,\ \forall\,y\in Y.\footnote{Since $f$ is exogenous, an equivalent formulation is to have the sender choose a conditional distribution $g_{x}(\cdot|\omega): \Omega \times X \to \Delta(Y)$ such that the $Y$-marginal aggregates to $g$.}
\end{equation*} 
The choice of $h$ is shared between the sender and receiver but hidden from the outside observer, and the posterior belief under $h$ is denoted $\mu_{h}(\cdot|x,y)$.

\paragraph{Receiver's action} Given the joint structure $h$, after observing the realized messages $(x,y)$, the receiver takes an action. Let $u: A \times \Omega \rightarrow \mathbb{R}$ denote the receiver's utility function. The receiver's optimal action satisfies 
\begin{equation*}
    a_{x,y} (u) \in \arg\max_{a \in A} \mathbb{E}_{\mu_h(\cdot|x,y)} [u(a,\omega)]. 
\end{equation*}
While the receiver's action $a_{x,y}(u)$ is publicly observed, her utility function $u$ is her \emph{private} information. We impose two restrictions on $u$ that are \emph{publicly} known: 
\begin{enumerate}[(i)]
    \item $a_0$ is the default action---unique optimal under the prior $\mu$, and
    \item the receiver's optimal action is monotone in the state.
\end{enumerate}
We capture this monotonicity formally using the \emph{single-crossing property} (SCP). A utility function $u: A \times \Omega \rightarrow \mathbb{R}$ satisfies the SCP if for any $a' > a$ and any $\omega' > \omega$,
\begin{equation}\label{eq:sc}
    u(a', \omega) \geq (>)\, u(a, \omega) \quad \implies \quad  u(a', \omega') \geq (>)\, u(a, \omega').
\end{equation}
Under the SCP, once the higher action $a'$ is preferred to $a$ at some state, it remains preferred at all higher states \citep{milgrom1994monotone}. Accordingly, it is public knowledge that $u$ belongs to the following class: 
\begin{equation*}
    \mathcal{U}:= \left \{ u: A \times \Omega \rightarrow \mathbb{R} \ : \ u ~\text{satisfies the SCP and}~ \{a_0\} = \arg\max_{a \in A} \mathbb{E}_\mu [u(a,\omega)]  \right \}. 
\end{equation*}

\paragraph{Information hierarchy}
Our model has three tiers of information, summarized in Figure \ref{fig:info_hierarchy}: what is \textbf{private} to the receiver (utility $u \in \mathcal{U}$), what is \textbf{shared} between sender and receiver (joint information structure $h$), and what is \textbf{public} (all other information).

\begin{figure}[h]
\begin{center}
			\begin{tikzpicture}
				\node (receiver) [draw, minimum width=10cm, minimum height=6cm, label={[anchor=south west, xshift=1em, yshift=-1.5em]north west: Private }] {};
				
				\node (sender_receiver) [draw, below=1.2cm of receiver.north west, anchor=north west, minimum width=9.5cm, minimum height=4.8cm, label={[anchor=south west, xshift=1em, yshift=-1.5em]north west: Shared}] {};
				
				\node (common) [draw, below=1.2cm of sender_receiver.north west, anchor=north west, minimum width=9cm, minimum height=3.6cm, label={[anchor=south west, xshift=1em, yshift=-1.7em]north west:Public}] {};
				
				\node at (common.west) [xshift=2.5cm, yshift=0.4cm]  {$f: \Omega \rightarrow \Delta (X) $};
				\node at (common.west) [xshift=2.5cm, yshift=-0.4cm]    {$g: \Omega \rightarrow \Delta (Y) $};
				\node at (common.west) [xshift=2.5cm, yshift=-1.2cm] {$\mu \in \Delta (\Omega)$};
				\node at (common.west) [xshift=7cm, yshift=1.2cm]  {$\mathcal{U}$};
				\node at (common.west) [xshift=7cm, yshift=0.4cm]    {$\Omega, \ A$};
				\node at (common.west) [xshift=7cm, yshift=-0.4cm] {$ x, \ y  $};
				\node at (common.west) [xshift=7cm, yshift=-1.2cm] {$ a_{x, y}(u)$};
				\node at (sender_receiver.north) [xshift=2cm, yshift=-.6cm] {$h: \Omega \rightarrow \Delta (X \times Y)$};
				\node at (receiver.north) [xshift=2cm, yshift=-.6cm] {$u$};
				
			\end{tikzpicture}
		\end{center}
        \caption{Hierarchy of Information}
        \label{fig:info_hierarchy}
    \end{figure}

\bigskip 
We refer to Section \ref{sec:discussion_assumption} for further discussion of our modeling assumptions, including the possibility of repeated message observations and the single-crossing restriction on the receiver's utility function.

\subsection{Secrecy and Plausible Deniability \label{subsec:S-PD}}

The sender's communication strategy---his choice of the joint information structure $h$---is subject to two constraints: secrecy and plausible deniability. These constraints aim to protect the communication from the surveillance of the outside observer.

\paragraph{Secrecy} Secrecy protects the content of the communication---it prevents the outside observer from learning about the state from the sender's message. Formally, this amounts to requiring that the marginal distribution of $y$ be independent of the state, following precisely the notion of perfect secrecy in \citet{shannon1949}.

\begin{definition}\label{def:secrecy}
    A joint information structure $h: \Omega \to \Delta(X \times Y)$ satisfies \textbf{secrecy} if the marginal distribution of $y$ is independent of the state $\omega$; that is,
    \begin{equation*}
        \sum_{x \in X} h(x,y|\omega) = g(y|\omega) = g(y), \quad \forall\, \omega \in \Omega, \ \forall\, y \in Y.
    \end{equation*}
\end{definition}
In words, secrecy requires the marginal information structure $g$ to be uninformative about the state. Consequently, any information the sender conveys to the receiver must operate through the dependence between $x$ and $y$ induced by the joint structure $h$.

\paragraph{Plausible deniability} Plausible deniability protects the communication from being revealed by the receiver's behavior---it prevents the outside observer from inferring the existence of hidden information from the receiver's action. Formally, this requires that every action the receiver takes under $h$ could also be justified as optimal given her baseline message $x$ alone, for some utility consistent with the publicly known restrictions.

\begin{definition}\label{def:plausible_deniability}
A joint information structure $h: \Omega \to \Delta(X \times Y)$ satisfies \textbf{plausible deniability} if, for every $u \in \mathcal{U}$ and $(x,y)$ in the support of $h$, there exists $\tilde{u} \in \mathcal{U}$ such that
\begin{equation*}
    a_{x,y}(u) \in \mathop{\arg\max}_{a \in A} \mathbb{E}_{\mu_f(\cdot |x)}[\tilde{u}(a,\omega)],
\end{equation*}
where $\mu_f(\cdot |x)$ denotes the posterior induced by the baseline structure $f$ upon observing $x$.
\end{definition}

This ensures that the outside observer, upon seeing the receiver's action, cannot conclude that she acted on information beyond her baseline message $x$. The receiver can always safely attribute her action to some preference $\tilde{u} \in \mathcal{U}$ and the baseline posterior $\mu_f(\cdot|x)$, without invoking the sender's message $y$ or its dependence on $x$. 

Crucially, plausible deniability is required to hold for \emph{all} $u \in \mathcal{U}$, rather than for a specific utility. This is a \emph{robustness} requirement: the sender designs communication to ensure that any action the receiver might take after observing $(x,y)$ admits a plausible-deniability justification, regardless of her true utility function. Equivalently, when the sender communicates with a population of receivers, plausible deniability is guaranteed uniformly across all possible preferences in $\mathcal{U}$.

\subsection{Sender's Communication Problem}
The sender aims to maximize the receiver's communication value across all utility functions in $\mathcal{U}$, subject to secrecy and plausible deniability. 
We define dominance as follows: a joint information structure $h$ \emph{dominates} $h'$ over $\mathcal{U}$ if, for every utility function $u \in \mathcal{U}$, the receiver's expected utility is greater under $h$ than under $h'$, i.e., 
\begin{align*}
    \mathbb{E} \left[ \max_{a \in A} \mathbb{E}_{\mu_h(\cdot |x,y)} [u(a,\omega)] \right] \geq \mathbb{E} \left[ \max_{a \in A} \mathbb{E}_{\mu_{h'}(\cdot |x,y)} [u(a,\omega)] \right], \quad \forall\, u \in \mathcal{U}.
\end{align*}
If the inequality holds strictly for at least one $u \in \mathcal{U}$, $h$ \textit{strictly dominates} $h'$ over $\mathcal{U}$. 
If $h $ dominates $h'$ and $h'$ dominates $h$, we call $h$ and $h'$ are \textit{equivalent}. 
For brevity, we omit the qualifier ``over $\mathcal{U}$'' hereafter. 

This dominance order is weaker than the standard Blackwell order, which requires the inequality to hold for \textit{all} utility functions, rather than a restricted class $\mathcal{U}$.\footnote{The Lehmann order \citep{lehmann1988comparing} refines the Blackwell order for monotone decision problems by assuming a one-dimensional signal space with a natural order (via the monotone likelihood ratio property). Since our goal is to characterize multi-dimensional joint information structures ($X \times Y$) lacking a natural order, we do not utilize the Lehmann order here.} 
While we identify cases where $h$ dominates $h'$ without being Blackwell more informative, the Blackwell order remains a useful sufficient condition: any $h'$ obtained by garbling $h$ is necessarily dominated by $h$. 

For any set $\mathcal{H}$ of joint information structures, 
\begin{enumerate}[(i)]
    \item \textbf{Frontier}: a structure $h\in\mathcal H$ lies on the \emph{frontier} of $\mathcal{H}$ if it is undominated in $\mathcal{H}$ (i.e., no $h'\in\mathcal{H}$ strictly dominates $h$).
    \item \textbf{Greatest Element}: A structure $h^{*}\in\mathcal{H}$ is \emph{greatest} in $\mathcal{H}$ if it dominates every $h\in\mathcal H$. If there exists a greatest structure in $\mathcal{H}$, it is unique up to equivalence. 
\end{enumerate}

Let $\mathcal H^{SPD}$ denote the set of joint information structures that satisfy both secrecy and plausible deniability. It is non-empty as it contains the independent structure $h(x,y|\omega) = f(x|\omega)g(y)$, under which $y$ conveys no information beyond what $x$ already provides. Thus, the sender's goal is well-defined: to choose a joint information structure in $\mathcal H^{SPD}$ that lies on the frontier of $\mathcal H^{SPD}$. Our main results characterize this frontier and provide conditions under which a greatest element of $\mathcal H^{SPD}$ exists and can be constructed in a tractable manner.

\section{Plausible Deniability and Secrecy}\label{sec:PD-secrecy}

In this section, we analyze the structural implications of secrecy and plausible deniability on the joint information structure $h$. 

\subsection{Characterization of Plausible Deniability}\label{sec:PD}

To analyze plausible deniability, we first define `rationalizability' over arbitrary information structure. 

\begin{definition}
    Given an information structure $\phi: \Omega \to \Delta(Z)$ and a message realization $z \in \text{supp}(\phi)$, an action $a \in A$ is \textbf{rationalizable} at $z$ (under $\phi$) if there exists $u \in \mathcal{U}$ such that $a$ is optimal under the posterior induced by observing $z$. Let $\mathcal{R}_\phi(z)\subseteq A$ denote the set of rationalizable actions at $z$. 
\end{definition}

We omit ``under $\phi$'' whenver the underlying information structure is clear from the context. The following lemma establishes the connection between plausible deniability and rationalizability. 

\begin{lemma} \label{lem:PD-R}
    A joint information structure $h$ satisfies plausible deniability if and only if $ \mathcal{R}_h(x,y) \subseteq \mathcal{R}_f (x)$ for all $(x,y) \in \text{supp}(h)$. 
\end{lemma}

Intuitively, when $h$ satisfies plausible deniability, any action in $\mathcal{R}_h(x,y)$ should also be rationalizable at $x$ under $f$, implying that it is in $\mathcal{R}_f(x)$. 
This lemma shows that characterizing plausible deniability reduces to understanding the set of rationalizable actions $\mathcal{R}$. The following lemma characterizes the set of rationalizable actions at arbitrary message. 

\medskip 

\begin{lemma}\label{lem:rationalizability}
    Let $\phi: \Omega \to \Delta(Z)$ be an information structure and $z \in \text{supp}(\phi)$ its realization.
    \begin{enumerate}[(i)]
        \item If $\phi(z|\omega)$ is decreasing and non-constant in $\omega$, $\mathcal{R}_\phi(z) =\{a \ : \  a \le a_0\} $.

        \item If $\phi(z|\omega)$ is increasing and non-constant in $\omega$, $\mathcal{R}_\phi(z) =\{a \ : \  a \ge a_0\} $.
            
        \item If $\phi(z|\omega)$ is constant in $\omega$, $\mathcal{R}_\phi(z) = \{a_0\} $;
        
        \item If $\phi(z|\omega)$ is non-monotone in $\omega$, $\mathcal{R}_\phi(z) = A$.
    \end{enumerate}
\end{lemma}

To provide intuition, first consider the case where $\phi(z|\omega)$ is decreasing in $\omega$. Then, observing $z$ shifts probability mass towards lower states relative to the prior. For any action $a > a_{0}$, the SCP implies that the payoff change from $a_0$ to $a$ is negative below some cutoff state and positive above it. Since $a_{0}$ is optimal under the prior, the prior already assigns sufficient weight to the lower (negative) region. Observing $z$ further shifts probability mass toward these lower states, reinforcing the dominance of $a_0$ and making $a > a_0$ impossible to rationalize. By contrast, the belief shift favors lower actions, and we show that every $a \le a_0$ can be rationalized at $z$ by some $u \in \mathcal{U}$, establishing $\mathcal{R}_\phi(z) = \{ a: a \le a_0 \} $. 
The case in which $\phi(z|\omega)$ is increasing is symmetric. 

When $\phi(z|\omega)$ is non-monotone, the likelihood is increasing over some states and decreasing over others, so observing $z$ shifts beliefs toward higher states in some regions and toward lower states in others. This creates enough flexibility to rationalize any action. For example, to rationalize $a > a_{0}$, one constructs a single-crossing utility function in which the payoff difference between $a$ and $a_0$ is concentrated on the states where the belief shift favors higher actions, while being negligible elsewhere. Under the prior, these states carry insufficient weight for $a$ to be more favorable than $a_0$, but after observing $z$, the upward belief shift at those states makes $a$ optimal. The symmetric argument rationalizes any $a < a_0$ by concentrating on states where the belief shift favors lower actions. Thus, all actions are rationalizable, establishing $\mathcal{R}_\phi(z) = A$. 

Lemma \ref{lem:rationalizability} applies to any information structure, and in particular to both baseline $f$ and joint structure $h$. Plausible deniability therefore requires that, for every $(x,y)$ in the support of $h$, the monotonicity of $h(x,y|\omega)$ in $\omega$ is compatible with that of $f(x|\omega)$. To state this precisely, we partition $X$ according to how $f(x|\omega)$ varies in $\omega$: 
\begin{enumerate}[(i)]
    \item let $D \subseteq X$ collect the messages $d$, for which $f(d|\omega)$ is decreasing in $\omega$, 
    \item let $I \subseteq X$ collect the messages $i$, for which $f(i|\omega)$ is increasing in $\omega$, and 
    \item let $S \subseteq X$ collect the remaining messages $s$, for which $f(s|\omega)$ is non-monotone in $\omega$.
\end{enumerate}
``Increasing'' and ``decreasing'' are understood in the weak sense, so a constant function qualifies as both. We break ties by assigning any $x$ with $f(x|\omega)$ constant in $\omega$ to $D$.

\begin{theorem}\label{thm:PD-characterization}
A joint information structure $h:\Omega\to\Delta(X\times Y)$ satisfies plausible deniability if and only if, for every $y\in Y$:
\begin{enumerate}[(i)]
    \item for each $d\in D$, $h(d,y|\omega)$ is decreasing in $\omega$, and is constant in $\omega$ whenever $f(d | \omega)$ is;
    \item for each $i\in I$, $h(i,y| \omega)$ is increasing in $\omega$; and
    \item for each $s\in S$, $h(s,y| \omega)$ is unrestricted.
\end{enumerate}
\end{theorem}

The proof of Theorem \ref{thm:PD-characterization} follows directly from Lemma \ref{lem:PD-R} and Lemma \ref{lem:rationalizability}. Let $\mathcal{H}^{PD}$ denote the set of joint information structures that satisfy plausible deniability. Theorem \ref{thm:PD-characterization} provides a complete characterization of $\mathcal{H}^{PD}$ in terms of the monotonicity of the likelihood $h(x,y|\omega)$ in $\omega$. 

A key feature of this characterization is that the restriction on $h(x,y|\omega)$ depends only on the monotonicity type of $f(x|\omega)$: plausible deniability constrains $h(x,y|\omega)$ independently across $x$. Therefore, the most informative joint information structure satisfying plausible deniability can be constructed message by message, choosing, for each $x$, the most informative decomposition $\{h(x,y|\omega)\}_{y\in Y}$ of $f(x|\omega)$ consistent with the monotonicity restriction.

For messages in $S$, there is no restriction, so the joint structure can fully reveal the state whenever such a message is observed---the most informative possible choice. For messages in $D$, viewing $h(d,y|\omega)$ as a vector in $\mathbb{R}^n$ indexed by $\omega$, the monotonicity restriction requires that each $h(d, y |\omega)$ must be a decreasing vector, and these vectors must sum to $f(d |\omega)$. The decreasing vectors form a convex cone whose extreme rays take a simple form: for each cutoff $k$, a step-down vector that is constant on $\{\omega_1, \ldots, \omega_k\}$ and zero on $\{\omega_{k+1}, \ldots, \omega_n\}$. Because every decreasing vector is a nonnegative combination of these extreme rays, any alternative decomposition of $f(d |\omega)$ can be recovered from the extreme-ray decomposition by garbling. The extreme-ray decomposition is therefore most informative. The case $i \in I$ is symmetric, with extreme rays that are zero on $\{\omega_1, \ldots, \omega_{k-1}\}$ and constant on $\{\omega_k, \ldots, \omega_n\}$.

We illustrate the construction with the following running example.
\begin{example}\label{ex:running}
    Let $\Omega = \{\omega_1, \omega_2, \omega_3\}$ and $X = \{d, s, i\}$, with baseline structure
    \begin{table}[H]
    \centering\footnotesize $f(x|\omega) =$ 
    \begin{tabular}{c|ccc}
         &  $d$ &$s$ &$i$\\
         \hline
         $\omega_1$ & $0.6$& $0.2$ &$0.2$\\ 
         $\omega_2$ & $0.3$& $0.3$ &$0.4$\\ 
         $\omega_3$ & $0.1$& $0.2$ &$0.7$\\ 
    \end{tabular}
    \end{table}
    \noindent 
    so that $f(d|\omega)$ is decreasing, $f(s|\omega)$ is non-monotone, and $f(i|\omega)$ is increasing. 
\end{example}

Given this baseline structure, a greatest element $\overline{h} \in \mathcal{H}^{PD}$ can be constructed message by message. Since $f(s|\omega)$ is non-monotone, we decompose the $s$-column to concentrate on single states, yielding full revelation whenever $s$ is observed. For $d$, we decompose $f(d|\omega) = (0.6, 0.3, 0.1)^{\intercal}$ into decreasing extreme rays, one for each cutoff $k \in \{1,2,3\}$; symmetrically, we decompose $f(i|\omega)^{\intercal} = (0.2, 0.4, 0.7)^{\intercal}$ into increasing extreme rays. Combining all three gives 
\begin{table}[H]
    \centering\footnotesize $\overline{h}(x, y|\omega) =$
    \begin{tabular}{c|ccc|ccc|ccc}
        & $(d, y_1)$ & $(d, y_2)$ & $(d, y_3)$
        & $(s, y_1)$ & $(s, y_2)$ & $(s, y_3)$
        & $(i, y_1)$ & $(i, y_2)$ & $(i, y_3)$ \\
        \hline
        $\omega_1$ & $0.3$ & $0.2$ & $0.1$
            & $0.2$ & $0$ & $0$
            & $0.2$ & $0$ & $0$ \\
        $\omega_2$ & $0$ & $0.2$ & $0.1$
            & $0$ & $0.3$ & $0$
            & $0.2$ & $0.2$ & $0$ \\
        $\omega_3$ & $0$ & $0$ & $0.1$
            & $0$ & $0$ & $0.2$
            & $0.2$ & $0.2$ & $0.3$ \\
    \end{tabular}
\end{table}
\noindent
The structure $\overline{h}$ satisfies plausible deniability, and any $h \in \mathcal{H}^{PD}$ can be obtained as a garbling of $\overline{h}$. To see why, note that each column of $\overline{h}$ is an extreme ray, so the only way to satisfy the monotonicity restrictions while preserving the column sums $f(x|\omega)$ is to produce a coarser decomposition---which is precisely a garbling. Hence $\overline{h}$ is a greatest element in $\mathcal{H}^{PD}$, and the construction extends to any baseline structure $f$. 

\begin{proposition}\label{prop:PD-greatest}
    A joint information structure $\overline{h} \in \mathcal{H}^{PD}$ is a greatest element of $\mathcal{H}^{PD}$ if and only if, for every $y \in Y$, each vector $\overline{h}(x, y|\omega)$ lies on an extreme ray of the relevant cone:
    \begin{enumerate}[(i)]
        \item for each $d \in D$: $\overline{h}(d, y|\omega)$ is constant on $\{\omega_1, \ldots, \omega_k\}$ and zero on $\{\omega_{k+1}, \ldots, \omega_n\}$ for some $k \in \{1, \ldots, n\}$;
        \item for each $i \in I$: $\overline{h}(i, y|\omega)$ is zero on $\{\omega_1, \ldots, \omega_{k-1}\}$ and constant on $\{\omega_k, \ldots, \omega_n\}$ for some $k \in \{1, \ldots, n\}$;
        \item for each $s \in S$: $\overline{h}(s, y|\omega)$ is concentrated on a single state, i.e., positive at some $\omega_k$ and zero elsewhere.
    \end{enumerate}
\end{proposition}

Notice that a PD-greatest structure $\overline{h}$ induces a simple posterior geometry: conditional on $(s, y)$ with $s \in S$, the state is fully revealed; conditional on $(d, y)$ with $d \in D$, the posterior is the prior conditioned on a lower tail $\{\omega_1, \ldots, \omega_k\}$; and conditional on $(i, y)$ with $i \in I$, the posterior is the prior conditioned on an upper tail $\{\omega_k, \ldots, \omega_n\}$.  However, PD-greatest structures need not satisfy secrecy: in Example~\ref{ex:running}, the $y$-marginal of $\overline{h}$ depends on $\omega$, so $\overline{h} \notin \mathcal{H}^{\text{SPD}}$. We therefore turn next to analyzing the implications of secrecy alone, to understand how it constrains the joint structure and how much of the PD-greatest posterior geometry can be preserved once secrecy is imposed.

\subsection{Secrecy and Signal Representation}\label{sec:secrecy_signal}

We start by introducing a family of joint information structures that satisfy secrecy, and then show that any secrecy-preserving structure is Blackwell dominated by a member of this family. This construction relies on the \textit{signal representation} of information structures introduced by \cite{green2022two}. Let $T=[0,1]$ be endowed with the Lebesgue measure $\lambda$. 

\begin{definition}[\cite{green2022two}]
    Given an information structure $f : \Omega \rightarrow \Delta (X)$ with a finite set $X$, a mapping $\psi: \Omega \times T \to  X$ is a \textbf{signal representation} of $f$ if $\psi(\omega, \cdot)$ is measurable for all $\omega\in\Omega$ and
    \begin{equation}
        \int_0^1 \mathbf{1}\{\psi(\omega, t) = x\} \, dt = f(x|\omega), \quad \forall\, x \in X, \ \omega \in \Omega.
        \label{eq:signal-marginal}
    \end{equation} 
\end{definition}

A signal representation naturally uses the continuum $T=[0,1]$, while the joint information structures in $\mathcal H$ have finite message spaces. We therefore first associate each representation $\psi$ with an intermediate joint object over $X\times T$. For each state $\omega$, define
\begin{equation*}
    \xi_\psi (x, t|\omega) := \mathbf{1} \{ \psi(\omega, t ) =x  \}.
\end{equation*}
Formally, $\xi_\psi$ is the density of a probability measure on $X\times T$ given by
\begin{equation*}
    \mathbb{P}_\omega(\{x\}\times T') := \lambda\big(\{t\in T':\psi(\omega,t)=x\}\big),
\end{equation*}
for every $x\in X$ and every measurable set $T'\subseteq T$. By \eqref{eq:signal-marginal}, the $X$-marginal of $\xi_\psi$ is exactly $f$. Furthermore, because $\sum_{x\in X} \xi_\psi(x,t|\omega) = \sum_{x\in X}\mathbf{1}\{\psi(\omega,t)=x\} = 1 $ for every $t\in T$, the $T$-marginal is uniform on $[0,1]$ and hence independent of the state. Thus, $\xi_\psi$ is a convenient intermediate object satisfying secrecy.

To further construct a joint information structure with a finite message space, we discretize the $T$-coordinate in the coarsest way that preserves the information contained in the representation. Since $\Omega$ and $X$ are finite, there are at most $|X|^{|\Omega|}$ distinct mappings from $\Omega$ to $X$. Partition $T$ into equivalence classes by grouping together $t$ and $t'$ whenever $\psi(\cdot,t)=\psi(\cdot,t')$. Associate each equivalence class with a distinct realization $y\in Y$, and denote the resulting finite partition of $T$ by $\{T_y\}_{y\in Y}$. Define
\begin{equation*}
    h_\psi(x,y|\omega) := \int_{T_y} \xi_\psi(x,t|\omega)\,dt = \int_{T_y}\mathbf{1}\{\psi(\omega,t)=x\}\,dt .
\end{equation*}
Then $h_\psi:\Omega\to\Delta(X\times Y)$ belongs to $\mathcal H$, and is Blackwell-equivalent to $\xi_\psi$. It inherits the baseline marginal condition from $\xi_\psi$, and it is secrecy-preserving because the marginal probability of each $y$ is
\begin{equation*}
    \sum_{x\in X}h_\psi(x,y|\omega)  = \int_{T_y}\sum_{x\in X}\xi_\psi(x,t|\omega)\,dt = \lambda(T_y),
\end{equation*}
which is independent of $\omega$.

\begin{definition}
A joint information structure $h:\Omega\to\Delta(X\times Y)$ is \textbf{signal-based} if there exists a signal representation $\psi:\Omega\times T\to X$ of $f$ such that $h = h_\psi$.
\end{definition}

\paragraph{A Graphical Illustration} Next, we provide a graphical illustration of the signal representation and the construction of signal-based joint information structures.

Intuitively, viewing each realization in $X$ as a distinct color, a signal $\psi$ ``paints'' the unit interval, creating a colored line for each state $\omega$. The condition in \eqref{eq:signal-marginal} simply requires that the total length painted with color $x$ in state $\omega$ equals $f(x|\omega)$. 

Consider the baseline structure $f$ in Example \ref{ex:running}. Figure \ref{fig:signal_representation} illustrates one possible signal $\psi$ representing this $f$. Each row corresponds to a state $\omega_i$, where the length of the segment assigned to $x$ exactly matches $f(x|\omega_i)$. As established above, this continuum object naturally reduces to a joint information structure. The interval $[0,1]$ is partitioned into five distinct equivalence classes, $y_1$ through $y_5$, determined by the vertical alignments of the segments. Within any region $y_k$, the mapping $\psi(\cdot,t)$ is identical for all $t \in y_k$. For instance, on $y_1 = [0,0.1)$, the signal assigns $\psi(\omega_i,t) = d$ for all three states. Thus, the signal-based joint information structure is implemented with a finite message space $Y = \{y_1, \cdots, y_5\}$. 

\begin{figure}[ht]
\centering
\begin{tikzpicture}[x=1cm,y=1cm]
  \def\hh{0.123} 

  \node[anchor=east] at (-0.4, 2.0) {$\omega_1$};
  \node[anchor=east] at (-0.4, 1.0) {$\omega_2$};
  \node[anchor=east] at (-0.4, 0.0) {$\omega_3$};

  \draw[line width=0.8pt] (0,2.5) -- (7.0,2.5);

  \draw[dashed, gray!70] (0.7, 2.5) -- (0.7, -0.5);
  \draw[dashed, gray!70] (2.1, 2.5) -- (2.1, -0.5);
  \draw[dashed, gray!70] (4.2, 2.5) -- (4.2, -0.5);
  \draw[dashed, gray!70] (5.6, 2.5) -- (5.6, -0.5);

  \node[anchor=south] at (0.35, 2.5) {\small $y_1$};
  \node[anchor=south] at (1.4, 2.5) {\small $y_2$};
  \node[anchor=south] at (3.15, 2.5) {\small $y_3$};
  \node[anchor=south] at (4.9, 2.5) {\small $y_4$};
  \node[anchor=south] at (6.3, 2.5) {\small $y_5$};

  \fill[orange!40!yellow] (0,2.0-\hh) rectangle (4.2,2.0+\hh);
  \fill[blue!30!white]    (4.2,2.0-\hh) rectangle (5.6,2.0+\hh);
  \fill[red!35!white]     (5.6,2.0-\hh) rectangle (7.0,2.0+\hh);
  \draw[line width=2pt, black] (4.2,2.0-\hh) -- (4.2,2.0+\hh);
  \draw[line width=2pt, black] (5.6,2.0-\hh) -- (5.6,2.0+\hh);
  \node at (2.1, 2.0) {\footnotesize $d$};
  \node at (4.9, 2.0) {\footnotesize $s$};
  \node at (6.3, 2.0) {\footnotesize $i$};

  \fill[orange!40!yellow] (0,1.0-\hh) rectangle (2.1,1.0+\hh);
  \fill[blue!30!white]    (2.1,1.0-\hh) rectangle (4.2,1.0+\hh);
  \fill[red!35!white]     (4.2,1.0-\hh) rectangle (7.0,1.0+\hh);
  \draw[line width=2pt, black] (2.1,1.0-\hh) -- (2.1,1.0+\hh);
  \draw[line width=2pt, black] (4.2,1.0-\hh) -- (4.2,1.0+\hh);
  \node at (1.05, 1.0) {\footnotesize $d$};
  \node at (3.15, 1.0) {\footnotesize $s$};
  \node at (5.6, 1.0) {\footnotesize $i$};

  \fill[orange!40!yellow] (0,0.0-\hh) rectangle (0.7,0.0+\hh);
  \fill[blue!30!white]    (0.7,0.0-\hh) rectangle (2.1,0.0+\hh);
  \fill[red!35!white]     (2.1,0.0-\hh) rectangle (7.0,0.0+\hh);
  \draw[line width=2pt, black] (0.7,0.0-\hh) -- (0.7,0.0+\hh);
  \draw[line width=2pt, black] (2.1,0.0-\hh) -- (2.1,0.0+\hh);
  \node at (0.35, 0.0) {\footnotesize $d$};
  \node at (1.4, 0.0) {\footnotesize $s$};
  \node at (4.55, 0.0) {\footnotesize $i$};
\end{tikzpicture}
\caption{A Signal Representation and Joint Information Structure}
\label{fig:signal_representation}
\end{figure}
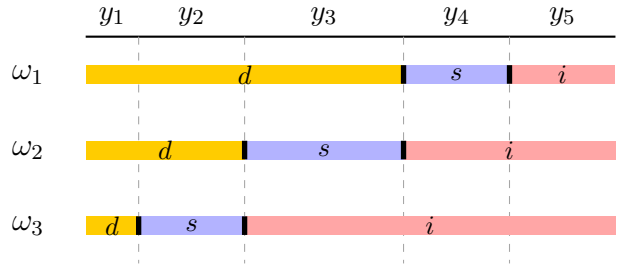

The specific arrangement in Figure \ref{fig:signal_representation} is only one of many possible representations. The sender can freely rearrange, split, or redistribute the segments within each row, as long as the total length assigned to each $x$ remains unchanged. Each such rearrangement induces a distinct signal-based joint information structure that preserves secrecy. Consequently, the sender's problem of choosing among signal-based structures is equivalent to choosing a partition of the unit interval for each state. Our next result shows that this restriction is without loss of generality for the sender's goal: for any secrecy-preserving structure, there exists a signal-based structure that Blackwell dominates it.

\begin{proposition}\label{prop:secrecy-frontier}
    For any joint information structure $h: \Omega \rightarrow \Delta(X \times Y)$ that satisfies secrecy, there exists a signal-based structure $\tilde{h}$ that is Blackwell more informative.
\end{proposition}

Proposition \ref{prop:secrecy-frontier} is closely related to Theorem 1 of \cite{strack2024}; see Section \ref{sec:relation_to_privacy} for a detailed comparison and for an explanation of how it can be obtained by reformulating our problem in their framework. We provide a direct proof to keep the argument self-contained and to express the reduction in signal-based terminology.

A signal-based structure satisfies secrecy, but the converse is not true: a general secrecy-preserving joint information structure need not arise from a signal representation. The intuition is easiest to see from the special case in which $Y$ is independent of $X$ and has a state-independent marginal distribution. Such a structure is secrecy-preserving, but need not be signal-based. A signal-based structure can preserve the sender's state-independent marginal distribution while coupling $Y$ with the baseline message through a signal representation. The sender's message remains uninformative on its own, but the joint observation $(X,Y)$ becomes more informative about the state. Proposition \ref{prop:secrecy-frontier} shows that this logic extends generally: every secrecy-preserving joint information structure is dominated by a signal-based one, so signal-based structures are sufficient for characterizing the secrecy frontier.

Consequently, identifying the maximally informative secrecy-preserving structure reduces to determining how the colors ($x$) and locations ($y$) of the segments can be arranged to maximize informativeness. In particular, if at every location on the interval the assigned color differs across all states, then observing the color and location together perfectly identifies $\omega$. Such a perfectly revealing arrangement is feasible if and only if the total length required for each color across all states does not exceed one (i.e., $\sum_\omega f(x| \omega) \le 1$). \citet{cavounidis2025} formalize this observation.

\begin{proposition}\label{prop:secrecy-fully-revealing}[\citealp[Theorem 1]{cavounidis2025}]
    There exists a secrecy-preserving joint information structure that fully reveals the state if and only if $\sum_{\omega} f(x|\omega) \leq 1$ for every $x \in X$.
\end{proposition}

When this condition fails, we show in Appendix \ref{app:secrecy-frontier-additional} that the coloring intuition yields a tractable characterization of the secrecy frontier in the binary-state case. However, with three or more states, the frontier may include structures that assign overlapping regions to a given color across states, even when a non-overlapping arrangement is feasible. This lack of structure makes a full characterization of the secrecy-only frontier difficult. In contrast, as we show in the next section, imposing plausible deniability in addition to secrecy substantially sharpens the structure of the problem and yields a tractable characterization of the frontier and optimal communication strategies.

\section{The Information Frontier and Optimal Communication}\label{sec:frontier_optimal}

In this section, we characterize the information frontier under both secrecy and plausible deniability, and show that it is attainable by a signal-based structure with a simple posterior geometry. We then provide conditions under which the optimal communication strategy, i.e., a greatest element of $\mathcal{H}^{SPD}$ exists and can be constructed in a simple way. 

\subsection{The Information Frontier}
Let $h:\Omega\to\Delta(X\times Y)$ be a signal-based joint information structure. By construction, for each $(x,y)$, the likelihood vector $h(x,y| \omega)$ takes only two values across states: either zero or some constant $c$. Consequently, satisfying plausible deniability requires that the support of $h(x,y| \omega)$ be aligned with the monotonicity of $f(x|\omega)$. 

In particular, for $x\in D\cup I$, i.e., when the plausible deniability constraint is active, $h(x,y|\omega)$ must lie on an extreme ray of the decreasing or increasing cone. Notice that the signal-based structure constructed in Figure \ref{fig:signal_representation} satisfies this property, and therefore constitutes a joint information structure that satisfies both secrecy and plausible deniability. In this following, we establish that every frontier point of $\mathcal H^{SPD}$ admits an equivalent signal-based representation satisfying this monotonicity restriction.

\begin{theorem}\label{thm:frontier_secrecy_PD}
Every $h\in\mathcal H^{SPD}$ is Blackwell dominated by some signal-based $\tilde h\in\mathcal H^{SPD}$.
\end{theorem}

The proof of Theorem \ref{thm:frontier_secrecy_PD} combines the insights from the characterization of plausible deniability with the signal representation of secrecy. Fix any $h \in \mathcal{H}^{SPD}$. By Proposition \ref{prop:secrecy-frontier}, there exists a signal-based structure $\tilde{h}$ that is Blackwell more informative than $h$. The remaining step is to show that $\tilde{h}$ can be chosen to satisfy the monotonicity restrictions in Theorem \ref{thm:PD-characterization}. The key observation is that, for each $x \in D \cup I$ and each $y$, the likelihood vector $h(x,y|\cdot)$ is monotone and hence admits a unique decomposition into extreme rays of the corresponding cone (similar to Proposition \ref{prop:PD-greatest}). By constructing $\tilde{h}$ so that each refined signal realization implements one such extreme ray, we preserve the original probability $h(x,y|\omega)$ while ensuring that $\tilde{h}$ satisfies the monotonicity restrictions. 

Theorem \ref{thm:frontier_secrecy_PD} implies that, when analyzing the information frontier under secrecy and plausible deniability, there is no loss of generality in restricting attention to signal-based structures. Indeed, let $h\in \mathcal H^{SPD}$ be a frontier point. By Theorem
\ref{thm:frontier_secrecy_PD}, there exists a signal-based $\tilde h\in \mathcal H^{SPD}$ that Blackwell dominates $h$, and hence $\tilde h$ dominates $h$. Because $h$ is a frontier point, this dominance cannot be strict. Therefore $h$ and $\tilde h$ are equivalent under the dominance order.

\begin{corollary}\label{cor:frontier_secrecy_PD}
For every frontier point $h\in \mathcal H^{SPD}$, there exists a signal-based structure $\tilde h\in \mathcal H^{SPD}$ such that $\tilde h$ is equivalent to $h$ under the dominance order.
\end{corollary}

Corollary \ref{cor:frontier_secrecy_PD} therefore establishes that, for the purpose of analyzing the frontier of $\mathcal H^{SPD}$, there is no loss in working with signal-based structures. Within this class, the posterior geometry is especially simple and closely parallels that of PD-greatest structures for directional messages. For messages in $D$ and $I$, each realized signal induces a likelihood vector lying on an extreme ray of the decreasing or increasing cone. Hence the induced posterior is the prior conditional on a lower or upper tail of states, respectively.

The main difference arises for non-monotone messages. In a PD-greatest structure, observing such a message can fully reveal the state. Under secrecy, however, full revelation need not remain feasible, because secrecy may require overlap in the segments assigned to the same non-monotone message across states. Consequently, within a signal-based frontier representative, the posterior following a non-monotone message need not be degenerate, although it is still given by the prior conditional on some subset of states. This observation suggests that tractability hinges on limiting the number of non-monotone messages. We therefore turn next to directional and almost-directional structures, where there is at most one such message.

\subsection{Almost-Directional Baseline Structures}

We next restrict attention to baseline information structures with at most one non-monotone message. 

\begin{definition}\label{def:directional_almost_directional}
The baseline information structure $f:\Omega\to\Delta(X)$ is \textbf{directional} if $S=\emptyset$, i.e., if every message is either increasing or decreasing in $\omega$. It is \textbf{almost-directional} if $|S| \leq 1$, i.e., if there is at most one non-monotone message.
\end{definition}

Almost-directional structures are useful tractable benchmarks, with directional structures corresponding to the special case of no non-monotone message. They capture baseline information structures that are largely ordered, while allowing a limited departure from monotonicity through a single exceptional message. Thus, restricting attention to at most one non-monotone message permits meaningful deviations from full monotonicity, while ruling out environments in which irregularities are spread across many messages. At the same time, the restriction encompasses important low-dimensional cases: when $|X| \leq 3$, any information structure satisfying the monotone likelihood ratio property is almost-directional.

For almost-directional baseline structures, the signal representation provides a tractable way to construct optimal communication. We focus on the following family of signal-based structures.

\begin{definition}\label{def:direction-ordered}
    A signal representation $\psi:\Omega\times T\to X$ of $f$ is \textbf{direction-ordered} if there exist functions $T_1,T_2:\Omega\to[0,1]$, with $T_1(\omega)\leq T_2(\omega)$ for all $\omega\in\Omega$, such that
    \begin{equation}
        \psi(\omega,t)\in
        \begin{cases}
            D, & \text{if } t<T_1(\omega),\\
            S, & \text{if } T_1(\omega)\leq t<T_2(\omega),\\
            I, & \text{if } T_2(\omega)\leq t.
        \end{cases}
    \end{equation}
\end{definition}

By likelihood monotonicity, both $T_1$ and $T_2$ are decreasing in $\omega$. The signal-based information structure depicted in Figure \ref{fig:signal_representation} is induced precisely by a direction-ordered representation. The following result shows that this structure is a greatest element of $\mathcal{H}^{SPD}$, and therefore an optimal communication for the baseline structure in Example \ref{ex:running}.

\begin{theorem}\label{thm:directional_almost_directional}
Suppose the baseline information structure $f$ is almost-directional. Then there exists a direction-ordered representation $\psi$ of $f$ such that $h_\psi$ is the greatest element of $\mathcal{H}^{SPD}$.
\end{theorem}

Theorem \ref{thm:directional_almost_directional} shows that, when the baseline structure is almost-directional, an optimal communication can be constructed explicitly and tractably. When $D$ and $I$ each contain multiple messages, it is enough to arrange the three classes in the order $D<S<I$ along the signal space. The individual messages within the $D$- and $I$-regions can then be assigned so that each satisfies the monotonicity condition required by plausible deniability. Figure \ref{fig:signal_representation_overlap} illustrates this class-level construction in a three-state example. The same idea extends to any finite state space.

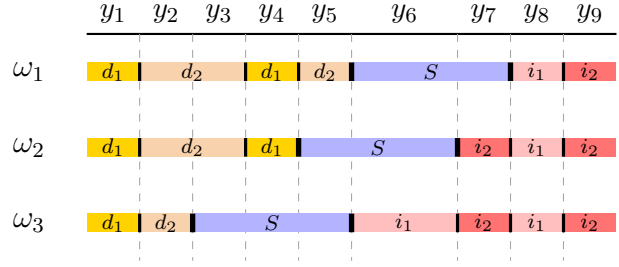
\begin{figure}[ht]
\centering
\begin{tikzpicture}[x=1cm,y=1cm]
  \def\hh{0.123}

  \node[anchor=east] at (-0.4, 2.0) {$\omega_1$};
  \node[anchor=east] at (-0.4, 1.0) {$\omega_2$};
  \node[anchor=east] at (-0.4, 0.0) {$\omega_3$};

  \draw[line width=0.8pt] (0,2.5) -- (7.0,2.5);

  \foreach \x in {0.7,1.4,2.1,2.8,3.5,4.9,5.6,6.3}
    \draw[dashed, gray!70] (\x, 2.5) -- (\x, -0.5);

  \node[anchor=south] at (0.35, 2.5) {\small $y_1$};
  \node[anchor=south] at (1.05, 2.5) {\small $y_2$};
  \node[anchor=south] at (1.75, 2.5) {\small $y_3$};
  \node[anchor=south] at (2.45, 2.5) {\small $y_4$};
  \node[anchor=south] at (3.15, 2.5) {\small $y_5$};
  \node[anchor=south] at (4.20, 2.5) {\small $y_6$};
  \node[anchor=south] at (5.25, 2.5) {\small $y_7$};
  \node[anchor=south] at (5.95, 2.5) {\small $y_8$};
  \node[anchor=south] at (6.65, 2.5) {\small $y_9$};


  \fill[orange!35!yellow] (0.0,2.0-\hh) rectangle (0.7,2.0+\hh);   
  \fill[orange!70!brown!35] (0.7,2.0-\hh) rectangle (2.1,2.0+\hh); 
  \fill[orange!35!yellow] (2.1,2.0-\hh) rectangle (2.8,2.0+\hh);   
  \fill[orange!70!brown!35] (2.8,2.0-\hh) rectangle (3.5,2.0+\hh); 
  \fill[blue!30!white] (3.5,2.0-\hh) rectangle (5.6,2.0+\hh);      
  \fill[red!25!white] (5.6,2.0-\hh) rectangle (6.3,2.0+\hh);       
  \fill[red!55!white] (6.3,2.0-\hh) rectangle (7.0,2.0+\hh);       

  \foreach \x in {0.7,2.1,2.8,6.3}
    \draw[line width=1.2pt] (\x,2.0-\hh) -- (\x,2.0+\hh);
  \foreach \x in {3.5,5.6}
    \draw[line width=2pt] (\x,2.0-\hh) -- (\x,2.0+\hh);

  \node at (0.35,2.0) {\scriptsize $d_1$};
  \node at (1.40,2.0) {\scriptsize $d_2$};
  \node at (2.45,2.0) {\scriptsize $d_1$};
  \node at (3.15,2.0) {\scriptsize $d_2$};
  \node at (4.55,2.0) {\scriptsize $S$};
  \node at (5.95,2.0) {\scriptsize $i_1$};
  \node at (6.65,2.0) {\scriptsize $i_2$};

  \fill[orange!35!yellow] (0.0,1.0-\hh) rectangle (0.7,1.0+\hh);   
  \fill[orange!70!brown!35] (0.7,1.0-\hh) rectangle (2.1,1.0+\hh); 
  \fill[orange!35!yellow] (2.1,1.0-\hh) rectangle (2.8,1.0+\hh);   
  \fill[blue!30!white] (2.8,1.0-\hh) rectangle (4.9,1.0+\hh);      
  \fill[red!55!white] (4.9,1.0-\hh) rectangle (5.6,1.0+\hh);       
  \fill[red!25!white] (5.6,1.0-\hh) rectangle (6.3,1.0+\hh);       
  \fill[red!55!white] (6.3,1.0-\hh) rectangle (7.0,1.0+\hh);       

  \foreach \x in {0.7,2.1,5.6,6.3}
    \draw[line width=1.2pt] (\x,1.0-\hh) -- (\x,1.0+\hh);
  \foreach \x in {2.8,4.9}
    \draw[line width=2pt] (\x,1.0-\hh) -- (\x,1.0+\hh);

  \node at (0.35,1.0) {\scriptsize $d_1$};
  \node at (1.40,1.0) {\scriptsize $d_2$};
  \node at (2.45,1.0) {\scriptsize $d_1$};
  \node at (3.85,1.0) {\scriptsize $S$};
  \node at (5.25,1.0) {\scriptsize $i_2$};
  \node at (5.95,1.0) {\scriptsize $i_1$};
  \node at (6.65,1.0) {\scriptsize $i_2$};

  \fill[orange!35!yellow] (0.0,0.0-\hh) rectangle (0.7,0.0+\hh);   
  \fill[orange!70!brown!35] (0.7,0.0-\hh) rectangle (1.4,0.0+\hh); 
  \fill[blue!30!white] (1.4,0.0-\hh) rectangle (3.5,0.0+\hh);      
  \fill[red!25!white] (3.5,0.0-\hh) rectangle (4.9,0.0+\hh);       
  \fill[red!55!white] (4.9,0.0-\hh) rectangle (5.6,0.0+\hh);       
  \fill[red!25!white] (5.6,0.0-\hh) rectangle (6.3,0.0+\hh);       
  \fill[red!55!white] (6.3,0.0-\hh) rectangle (7.0,0.0+\hh);       

  \foreach \x in {0.7,4.9,5.6,6.3}
    \draw[line width=1.2pt] (\x,0.0-\hh) -- (\x,0.0+\hh);
  \foreach \x in {1.4,3.5}
    \draw[line width=2pt] (\x,0.0-\hh) -- (\x,0.0+\hh);

  \node at (0.35,0.0) {\scriptsize $d_1$};
  \node at (1.05,0.0) {\scriptsize $d_2$};
  \node at (2.45,0.0) {\scriptsize $S$};
  \node at (4.20,0.0) {\scriptsize $i_1$};
  \node at (5.25,0.0) {\scriptsize $i_2$};
  \node at (5.95,0.0) {\scriptsize $i_1$};
  \node at (6.65,0.0) {\scriptsize $i_2$};

\end{tikzpicture}
\caption{A direction-ordered signal representation with multiple messages in $D$ and $I$}
\label{fig:signal_representation_overlap}
\end{figure}

\subsubsection{Sketch of the Proof of Theorem \ref{thm:directional_almost_directional}}

We illustrate the proof of Theorem \ref{thm:directional_almost_directional} in the case of three states and three signals $X = \{d, s, i\}$, as in Example \ref{ex:running}, where $d$ and $i$ are the monotone messages and $s$ is the unique non-monotone message. By Theorem \ref{thm:frontier_secrecy_PD}, it suffices to restrict attention to signal-based structures satisfying the monotonicity requirements imposed by plausible deniability. In particular, for each $y$, if a higher state is assigned message $d$, then every lower state must also be assigned $d$; symmetrically, if a lower state is assigned $i$, then every higher state must also be assigned $i$. Moreover, after a relabeling of the $y$-axis, which is Blackwell equivalent, we may assume without loss of generality that the $d$-regions are left-aligned across states. This rearrangement is feasible as $f(d|\omega)$ is decreasing in $\omega$.

It therefore remains to show that, within each state, whenever an $i$-region lies to the left of an $s$-region, the two regions can be repositioned so as to yield a joint structure that dominates the original one. We consider three cases, according to the state in which this misalignment occurs. Since the regions are induced by Lebesgue-measurable subsets of $T$, there is no loss in the arguments below by treating them as intervals of equal measure.

\paragraph{Misalignment in $\omega_1$.} Suppose that, in $\omega_1$, some $i$-region lies to the left of some $s$-region. Then a relabeling of the corresponding messages yields a Blackwell-equivalent structure in which that $i$-region is moved to the right of that $s$-region. Hence any misalignment in $\omega_{1}$ can be removed without changing informativeness. 

\paragraph{Misalignment in $\omega_2$.} Suppose that, in $\omega_2$, some $i$-region $y$ lies to the left of some $s$-region $y'$, with $y<y'$. Plausible deniability implies that $y$ must also be assigned $i$ in $\omega_3$, while $y'$ must be assigned $s$ in $\omega_1$ and either $s$ or $i$ in $\omega_3$. Since $\omega_1$ has already been aligned, $y$ must be assigned either $d$ or $s$ in $\omega_1$. We illustrate the configuration shown in Figure \ref{fig:reverse_garbling}; the remaining configurations follow by either a similar argument, or a relabeling as in the previous case. 

Given the original joint structure in Figure \ref{fig:reverse_garbling}, now construct a new joint structure by exchanging the $i$- and $s$-regions in $\omega_2$. This is feasible because the two regions have equal measure. The resulting structure is more informative in the Blackwell sense. Indeed, starting from the new structure, one can recover the original one by a garbling that relabels $(s, y)$ as $(s, y')$, $(i, y)$ as $(i, y')$, and $(i, y')$ as $(i, y)$, leaving all other realizations unchanged. 

\begin{figure}[ht]
\centering
\begin{tikzpicture}[x=1.2cm,y=1cm]
\def\hh{0.123}

  \node at (1.5,3.1) {\small Original};

  \node[anchor=east] at (-0.3, 2.0) {$\omega_1$};
  \node[anchor=east] at (-0.3, 1.0) {$\omega_2$};
  \node[anchor=east] at (-0.3, 0.0) {$\omega_3$};

  \draw[line width=0.8pt] (0,2.5) -- (3,2.5);
  \draw[dashed, gray!70] (1.5,2.5) -- (1.5,-0.5);

  \node[anchor=south] at (0.75,2.5) {\small $y$};
  \node[anchor=south] at (2.25,2.5) {\small $y'$};

  \draw[line width=7pt, color=orange!40!yellow] (0,2.0) -- (1.5,2.0);
  \draw[line width=7pt, color=blue!30!white] (1.5,2.0) -- (3,2.0);
  \draw[line width=2pt, black] (1.5,2.0-\hh) -- (1.5,2.0+\hh);

  \draw[line width=7pt, color=red!35!white] (0,1.0) -- (1.5,1.0);
  \draw[line width=7pt, color=blue!30!white] (1.5,1.0) -- (3,1.0);
  \draw[line width=2pt, black] (1.5, 1.0-\hh) -- (1.5, 1.0+\hh);

  \draw[line width=7pt, color=red!35!white] (0,0.0) -- (1.5,0.0);
  \draw[line width=7pt, color=red!35!white] (1.5,0.0) -- (3,0.0);
  \draw[line width=2pt, black] (1.5, 0-\hh) -- (1.5, 0+\hh);

  \node at (0.75,2.0) {\small $d$};
  \node at (0.75,1.0) {\small $i$};
  \node at (0.75,0.0) {\small $i$};

  \node at (2.25,2.0) {\small $s$};
  \node at (2.25,1.0) {\small $s$};
  \node at (2.25,0.0) {\small $i$};


  \node at (7.0,3.1) {\small New};

  \node[anchor=east] at (5.2, 2.0) {$\omega_1$};
  \node[anchor=east] at (5.2, 1.0) {$\omega_2$};
  \node[anchor=east] at (5.2, 0.0) {$\omega_3$};

  \draw[line width=0.8pt] (5.5,2.5) -- (8.5,2.5);
  \draw[dashed, gray!70] (7.0,2.5) -- (7.0,-0.5);

  \node[anchor=south] at (6.25,2.5) {\small $y$};
  \node[anchor=south] at (7.75,2.5) {\small $y'$};

  \draw[line width=7pt, color=orange!40!yellow] (5.5,2.0) -- (7.0,2.0);
  \draw[line width=7pt, color=blue!30!white] (7.0,2.0) -- (8.5,2.0);
  \draw[line width=2pt, black] (7.0, 2.0-\hh) -- (7.0, 2.0+\hh);

  \draw[line width=7pt, color=blue!30!white] (5.5,1.0) -- (7.0,1.0);
  \draw[line width=7pt, color=red!35!white] (7.0,1.0) -- (8.5,1.0);
   \draw[line width=2pt, black] (7.0, 1.0-\hh) -- (7.0, 1.0+\hh);

  \draw[line width=7pt, color=red!35!white] (5.5,0.0) -- (7.0,0.0);
  \draw[line width=7pt, color=red!35!white] (7.0,0.0) -- (8.5,0.0);
  \draw[line width=2pt, black] (7.0, 0-\hh) -- (7.0, 0+\hh);

  \node at (6.25,2.0) {\small $d$};
  \node at (6.25,1.0) {\small $s$};
  \node at (6.25,0.0) {\small $i$};

  \node at (7.75,2.0) {\small $s$};
  \node at (7.75,1.0) {\small $i$};
  \node at (7.75,0.0) {\small $i$};

\end{tikzpicture}
\caption{A case of misalignment in $\omega_2$}
\label{fig:reverse_garbling}
\end{figure}
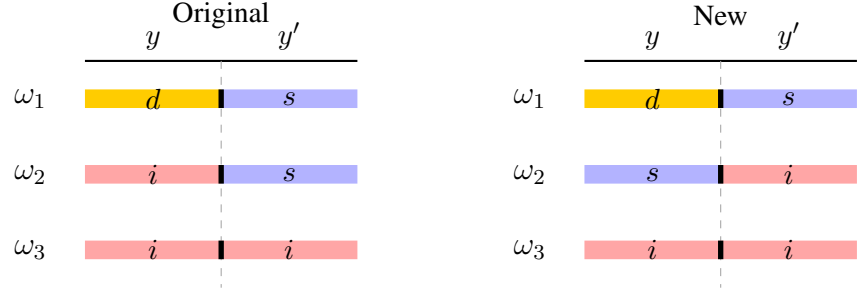

\paragraph{Misalignment in $\omega_3$.} Suppose that, in $\omega_3$, some $i$-region $y$ lies to the left of some $s$-region $y'$, with $y < y'$. Plausible deniability implies that $y'$ must be assigned either $d$ or $s$ in both $\omega_1$ and $\omega_2$. Since $\omega_1$ and $\omega_2$ are already aligned, it follows that $y$ must likewise be assigned either $d$ or $s$ in those states. We focus on the configuration shown in Figure \ref{fig:swapping} because it requires a dominance argument tailored to single-crossing preferences, rather than standard Blackwell dominance as in the previous cases. The remaining configurations can be handled by Blackwell arguments analogous to those used above.

Given the original joint structure in Figure \ref{fig:swapping}, construct a new joint structure by exchanging the $i$- and $s$-regions in state $\omega_3$. The only difference arises when $s$ is observed: in the original structure, observing $s$ either isolates the intermediate state $\omega_2$ or is uninformative, whereas in the new structure it induces a cutoff comparison around $\omega_2$. For arbitrary preferences, either form of information could be more valuable, so the two structures are not Blackwell comparable. Nevertheless, we show next that the new structure dominates the original one, in the sense that it yields a higher expected payoff for every $u\in\mathcal U$. The key intuition is that single crossing preserves the ranking of actions in the direction of the state: an upward action adjustment that is valuable at an intermediate state remains valuable at higher states. Thus, the new structure can relocate such an adjustment from $\omega_2$ to $\omega_3$, where it remains valuable by single crossing. The case of a downward action adjustment is symmetric.

\begin{figure}[ht]
\centering
\begin{tikzpicture}[x=1.2cm,y=1cm]

  \def\hh{0.123}

  \node at (1.5,3.1) {\small Original};

  \node[anchor=east] at (-0.3, 2.0) {$\omega_1$};
  \node[anchor=east] at (-0.3, 1.0) {$\omega_2$};
  \node[anchor=east] at (-0.3, 0.0) {$\omega_3$};

  \draw[line width=0.8pt] (0,2.5) -- (3,2.5);
  \draw[dashed, gray!70] (1.5,2.5) -- (1.5,-0.5);

  \node[anchor=south] at (0.75,2.5) {\small $y$};
  \node[anchor=south] at (2.25,2.5) {\small $y'$};

  \draw[line width=7pt, color=orange!40!yellow] (0,2.0) -- (1.5,2.0);
  \draw[line width=7pt, color=blue!30!white] (1.5,2.0) -- (3,2.0);
  \draw[line width=2pt, black] (1.5,2.0-\hh) -- (1.5,2.0+\hh);

  \draw[line width=7pt, color=blue!30!white] (0,1.0) -- (1.5,1.0);
  \draw[line width=7pt, color=blue!30!white] (1.5,1.0) -- (3,1.0);
  \draw[line width=2pt, black] (1.5,1.0-\hh) -- (1.5,1.0+\hh);

  \draw[line width=7pt, color=red!35!white] (0,0.0) -- (1.5,0.0);
  \draw[line width=7pt, color=blue!30!white] (1.5,0.0) -- (3,0.0);
  \draw[line width=2pt, black] (1.5, 0-\hh) -- (1.5, 0+\hh);

  \node at (0.75,2.0) {\small $d$};
  \node at (0.75,1.0) {\small $s$};
  \node at (0.75,0.0) {\small $i$};

  \node at (2.25,2.0) {\small $s$};
  \node at (2.25,1.0) {\small $s$};
  \node at (2.25,0.0) {\small $s$};


  \node at (7.0,3.1) {\small New};

  \node[anchor=east] at (5.2, 2.0) {$\omega_1$};
  \node[anchor=east] at (5.2, 1.0) {$\omega_2$};
  \node[anchor=east] at (5.2, 0.0) {$\omega_3$};

  \draw[line width=0.8pt] (5.5,2.5) -- (8.5,2.5);
  \draw[dashed, gray!70] (7.0,2.5) -- (7.0,-0.5);

  \node[anchor=south] at (6.25,2.5) {\small $y$};
  \node[anchor=south] at (7.75,2.5) {\small $y'$};

  \draw[line width=7pt, color=orange!40!yellow] (5.5,2.0) -- (7.0,2.0);
  \draw[line width=7pt, color=blue!30!white] (7.0,2.0) -- (8.5,2.0);
  \draw[line width=2pt, black] (7.0,2.0-\hh) -- (7.0,2.0+\hh);

  \draw[line width=7pt, color=blue!30!white] (5.5,1.0) -- (7.0,1.0);
  \draw[line width=7pt, color=blue!30!white] (7.0,1.0) -- (8.5,1.0);
  \draw[line width=2pt, black] (7.0,1.0-\hh) -- (7.0,1.0+\hh);

  \draw[line width=7pt, color=blue!30!white] (5.5,0.0) -- (7.0,0.0);
  \draw[line width=7pt, color=red!35!white] (7.0,0.0) -- (8.5,0.0);
  \draw[line width=2pt, black] (7.0, 0-\hh) -- (7.0, 0+\hh);

  \node at (6.25,2.0) {\small $d$};
  \node at (6.25,1.0) {\small $s$};
  \node at (6.25,0.0) {\small $s$};

  \node at (7.75,2.0) {\small $s$};
  \node at (7.75,1.0) {\small $s$};
  \node at (7.75,0.0) {\small $i$};
\end{tikzpicture}
\caption{A case of misalignment in $\omega_3$}
\label{fig:swapping}
\end{figure}
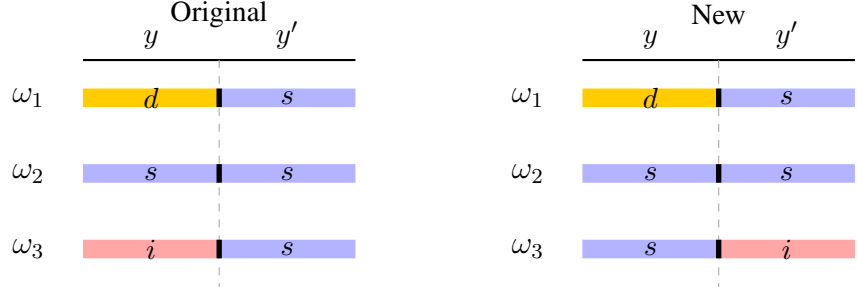

Fix any receiver's utility function $u \in \mathcal{U}$. Under the original structure, let $a_{d}$, $a_{s}$, and $a_{i}$ denote the receiver's optimal actions following messages $(d, y)$, $(s, y)$, and $(i, y)$, respectively. The receiver's optimal action following message $(s, y')$ is $a_{0}$, since her posterior coincides with the prior. Suppose $a_{s} \geq a_{0}$; the opposite case is symmetric.

Under the new structure, consider the following action rule for the receiver: choose $a_{d}$ following $(d, y)$, $a_{s}$ following $(s, y)$, $a_{0}$ following $(s, y')$, $a_{i}$ following $(i, y)$, and the same actions as in the original structure for all other realizations. Because $y$ and $y'$ have the same measure, the receiver's expected payoff under this action rule in the new structure is the same as under the optimal rule in the original structure for every state except $\omega_3$. Thus, the comparison depends only on the following difference:
\begin{align*}
 u(a_{s}, \omega_{3}) - u(a_{0}, \omega_{3}).
\end{align*}

Since $a_{s}$ is optimal for the receiver in state $\omega_{2}$, we have $u(a_{s}, \omega_{2}) \geq u(a_{0}, \omega_{2})$. By single crossing, this implies $u(a_{s}, \omega_{3}) \geq u(a_{0}, \omega_{3})$. Hence, the new structure yields a higher expected payoff for the receiver under the action rule described above, and therefore also under the optimal rule. Since this holds for every $u\in\mathcal U$, the new structure dominates the original one. 

\bigskip 

Repeating these operations eventually produces a direction-ordered structure, since all $i$-regions are moved to the right of all $s$-regions. The proof of Theorem \ref{thm:directional_almost_directional} shows that the same logic extends to more than three states through an induction over states. Once the regions have been ordered for lower states, any remaining misalignment at the next state can be removed by swapping equal-measure pieces. Although the possible configurations at an intermediate state are more varied than in the three-state example, they are all handled by the same three types of arguments illustrated above: relabeling, Blackwell-improving rearrangement, or a dominance-improving rearrangement justified by single crossing. Starting from an arbitrary $h \in \mathcal{H}^{SPD}$, iterating this induction yields a direction-ordered representation that dominates the original structure.

\subsection{Beyond Almost-Directional Baseline Structures}

When the baseline structure contains more than one non-monotone message, plausible deniability imposes no restrictions on those messages. As a result, the same intractability that arises under secrecy alone reemerges as the central bottleneck. In particular, the key swapping operation that yields a dominating structure in the almost-directional case need not yield a dominating structure when multiple non-monotone messages are present.

Consider the structure in Figure \ref{fig:swapping_counterexample}, which contains two non-monotone messages, $s_1$ and $s_2$. To obtain a direction-ordered structure, one would need to swap the $i$- and $s_1$-regions in state $\omega_3$. However, the resulting structure needs not dominate the original one. To see this, let $a_i$ denote the receiver's optimal action in state $\omega_i$, and let $a$ denote the receiver's optimal action following message $(s_1,y)$ under the new structure. Consider the following action rule under the original structure: choose the optimal actions following $(\cdot,y)$, $a_1$ following $(s_1,y')$, $a_2$ following $(s_2,y')$, and the same actions as in the new structure for all other realizations. The receiver's expected utility under this action rule in the original structure exceeds her payoff under the optimal rule in the new structure if and only if
\begin{align*}
\mu(\omega_{2})\big[ u(a_{2}, \omega_{2}) - u(a, \omega_{2})\big] 
> 
\mu(\omega_{1})\big[u(a_{1}, \omega_{1}) - u(a, \omega_{1})\big].
\end{align*}
This inequality can hold under single crossing when the payoff loss from choosing $a$ rather than the statewise optimum $a_2$ is large in $\omega_2$, while the loss from choosing $a$ rather than $a_1$ is comparatively small in $\omega_1$. Single crossing permits such cardinal differences as long as the induced action rankings move monotonically with the state. Hence, the swapping operation that leads to a direction-ordered structure in the almost-directional case need not yield a dominating structure here.

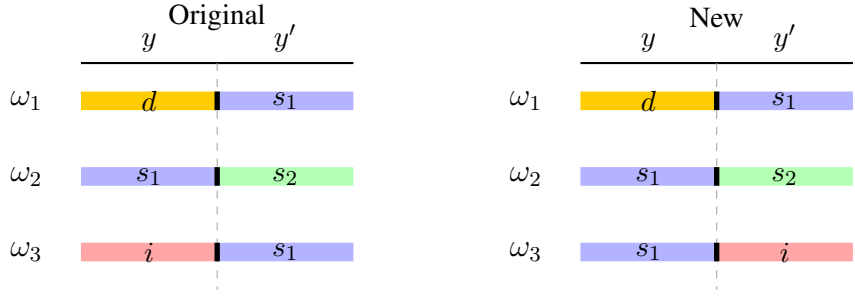
\begin{figure}[ht]
\centering
\begin{tikzpicture}[x=1.2cm,y=1cm]
\def\hh{0.123}

  \node at (1.5,3.1) {\small Original};

  \node[anchor=east] at (-0.3, 2.0) {$\omega_1$};
  \node[anchor=east] at (-0.3, 1.0) {$\omega_2$};
  \node[anchor=east] at (-0.3, 0.0) {$\omega_3$};

  \draw[line width=0.8pt] (0,2.5) -- (3,2.5);
  \draw[dashed, gray!70] (1.5,2.5) -- (1.5,-0.5);

  \node[anchor=south] at (0.75,2.5) {\small $y$};
  \node[anchor=south] at (2.25,2.5) {\small $y'$};

  \draw[line width=7pt, color=orange!40!yellow] (0,2.0) -- (1.5,2.0);
  \draw[line width=7pt, color=blue!30!white] (1.5,2.0) -- (3,2.0);
\draw[line width=2pt, black] (1.5,2.0-\hh) -- (1.5,2.0+\hh);

  \draw[line width=7pt, color=blue!30!white] (0,1.0) -- (1.5,1.0);
  \draw[line width=7pt, color=green!30!white] (1.5,1.0) -- (3,1.0);
 \draw[line width=2pt, black] (1.5, 1.0-\hh) -- (1.5, 1.0+\hh);

  \draw[line width=7pt, color=red!35!white] (0,0.0) -- (1.5,0.0);
  \draw[line width=7pt, color=blue!30!white] (1.5,0.0) -- (3,0.0);
   \draw[line width=2pt, black] (1.5, 0-\hh) -- (1.5, 0+\hh);

  \node at (0.75,2.0) {\small $d$};
  \node at (0.75,1.0) {\small $s_1$};
  \node at (0.75,0.0) {\small $i$};

  \node at (2.25,2.0) {\small $s_1$};
  \node at (2.25,1.0) {\small $s_2$};
  \node at (2.25,0.0) {\small $s_1$};


  \node at (7.0,3.1) {\small New};

  \node[anchor=east] at (5.2, 2.0) {$\omega_1$};
  \node[anchor=east] at (5.2, 1.0) {$\omega_2$};
  \node[anchor=east] at (5.2, 0.0) {$\omega_3$};

  \draw[line width=0.8pt] (5.5,2.5) -- (8.5,2.5);
  \draw[dashed, gray!70] (7.0,2.5) -- (7.0,-0.5);

  \node[anchor=south] at (6.25,2.5) {\small $y$};
  \node[anchor=south] at (7.75,2.5) {\small $y'$};

  \draw[line width=7pt, color=orange!40!yellow] (5.5,2.0) -- (7.0,2.0);
  \draw[line width=7pt, color=blue!30!white] (7.0,2.0) -- (8.5,2.0);
  \draw[line width=2pt, black] (7.0,2.0-\hh) -- (7.0,2.0+\hh);

  \draw[line width=7pt, color=blue!30!white] (5.5,1.0) -- (7.0,1.0);
  \draw[line width=7pt, color=green!30!white] (7.0,1.0) -- (8.5,1.0);
  \draw[line width=2pt, black] (7.0, 1.0-\hh) -- (7.0, 1.0+\hh);

  \draw[line width=7pt, color=blue!30!white] (5.5,0.0) -- (7.0,0.0);
  \draw[line width=7pt, color=red!35!white] (7.0,0.0) -- (8.5,0.0);
  \draw[line width=2pt, black] (7.0, 0-\hh) -- (7.0, 0+\hh);

  \node at (6.25,2.0) {\small $d$};
  \node at (6.25,1.0) {\small $s_1$};
  \node at (6.25,0.0) {\small $s_1$};

  \node at (7.75,2.0) {\small $s_1$};
  \node at (7.75,1.0) {\small $s_2$};
  \node at (7.75,0.0) {\small $i$};
\end{tikzpicture}
\caption{Counter example with multiple non-directional messages}
\label{fig:swapping_counterexample}
\end{figure}

The previous discussion illustrates how the presence of multiple non-monotone messages can complicate the problem. Nevertheless, the joint structure in Figure \ref{fig:signal_representation} remains optimal even if the message $s$ is further split into multiple non-monotone messages. The key observation is that, after the $D$- and $I$-regions are left- and right-aligned, respectively, the remaining non-monotone regions do not overlap across states. As a result, observing any non-monotone message together with its location fully reveals the state, regardless of how many non-monotone messages there are. The induced structure therefore coincides with the greatest structure satisfying plausible deniability alone. In other words, this is a case in which the PD-greatest structure can be implemented while preserving secrecy. The next theorem provides a sufficient condition on the baseline structure for this to happen.

\begin{theorem}\label{thm:PD_greatest_if}
    The PD-greatest element in $\mathcal H^{PD}$ can be implemented by a secrecy-preserving joint information structure, and hence is also a greatest element of $\mathcal H^{SPD}$, if for all $1 < k < n$, 
    \begin{equation*}
    \sum_{s\in S}f(s|\omega_k) \leq \min\left\{\sum_{d\in D}\left[f(d|\omega_{k - 1}) - f(d|\omega_k)\right], \,\sum_{i\in I}\left[f(i|\omega_{k + 1}) - f(i|\omega_k)\right]\right\}.
\end{equation*}
\end{theorem}

The condition in Theorem \ref{thm:PD_greatest_if} has a simple interpretation. At each interior state $\omega_k$, the total mass of non-monotone messages must fit within the local slack created by the monotone classes on either side. This slack is determined by the decrease in $D$-mass from $\omega_{k-1}$ to $\omega_k$ and the increase in $I$-mass from $\omega_k$ to $\omega_{k+1}$. Thus, the PD-greatest structure remains feasible whenever non-monotone messages are sufficiently sparse relative to the directional variation in the monotone classes.

This result can be viewed as a counterpart to Theorem 1 in \cite{cavounidis2025} (Proposition \ref{prop:secrecy-fully-revealing}) under the additional constraint of plausible deniability. There, full revelation under secrecy is feasible if and only if, for each message $x$, the total mass $\sum_{\omega} f(x|\omega)$ does not exceed one, so that the realizations of $x$ can be made disjoint across states. Here, rather than requiring global disjointness for every message, it is enough that the non-monotone regions can be separated across states using the local slack generated by the monotone classes. If the condition fails, the PD-greatest structure need not be implementable by this construction. Nevertheless, when there are multiple non-monotone messages, it may still be possible to arrange them across the signal space so that they do not overlap across states. This possibility points to a richer combinatorial problem governing how non-monotone messages can be placed while preserving both secrecy and plausible deniability. A full characterization of that problem is beyond the scope of the present paper.

\section{Discussions}\label{sec:discussions}

\subsection{Relation to Privacy Constraints and Independence Restrictions}\label{sec:relation_to_privacy}
Our secrecy constraint is closely related to the privacy constraints studied in \cite{strack2024} and \cite{he2026private}. All three papers characterize the maximally informative information structures subject to independence restrictions. The formal similarity is that, in each case, some random variable must remain statistically independent of another. The difference lies in which variables are protected and in how informativeness is evaluated.

To see the connection with \cite{strack2024}, fix a joint information structure $h:\Omega\to\Delta(X\times Y)$ whose marginal on $X$ is the baseline information structure $f$. Equivalently, write $h(x,y|\omega)=g_x(y|\omega)f(x|\omega)$, where $g_x(\cdot|\omega)$ is the conditional distribution over the sender's message given the state $\omega$ and the baseline message $x$. The secrecy constraint can then be written as $\sum_{x\in X} g_x(y|\omega)f(x|\omega)=g(y)$ for all $y$ and $\omega$. 

This constraint can be embedded into the framework of \cite{strack2024} by enlarging the state space to $\widetilde{\Omega}=\Omega\times X$ with prior $\widetilde{\mu}(\omega,x)=\mu(\omega)f(x|\omega)$. The sender's message is then generated from $\widetilde{\Omega}$ according to the kernel $g_x(\cdot|\omega)$, where $x$ is the baseline-message component of the enlarged state. If the protected attribute is $\theta(\omega,x)=\omega$, then privacy preservation with respect to $\theta$ is equivalent to our secrecy requirement.

The key difference from \cite{strack2024} is the informativeness criterion. Under this embedding, \cite{strack2024} ask how informative the privacy-preserving message $y$ can be about the enlarged state $(\omega,x)$, and compare such messages using the Blackwell order over $\widetilde{\Omega}$. In our model, by contrast, the evaluated object is the joint observation $(x,y)$, and its informativeness is measured with respect to only the state $\omega$. This distinction is substantive. If one restricted the decision problems in \cite{strack2024} to payoffs that depend only on the protected attribute $\theta(\omega,x)=\omega$, then all privacy-preserving messages would be equivalent, since they leave beliefs about $\theta$ unchanged. In our model, even when the receiver's payoff depends only on $\omega$, the design problem remains nontrivial because different choices of the dependence between $X$ and $Y$ change their joint informativeness.

The embedding also clarifies the connection between Proposition \ref{prop:secrecy-frontier} and Theorem 1 of \cite{strack2024}. In the enlarged state space, the non-protected component is the baseline message $x$, whose conditional distribution given the protected attribute $\omega$ is $f(\cdot|\omega)$. A Green--Stokey signal representation $\psi:\Omega\times T\to X$ represents exactly this conditional distribution. Thus, the location $t\in T$ plays the same role as the conditional quantile of the non-protected component in the reordered-quantile construction of \cite{strack2024}. Their Theorem 1 shows that every privacy-preserving signal is Blackwell dominated by such a reordered-quantile signal. Under the embedding above, this implies our Proposition \ref{prop:secrecy-frontier}, since the Blackwell order they use over $\widetilde{\Omega}$ is stronger than the criterion used here.

The relation to \cite{he2026private} can be stated most clearly in terms of the independence restrictions imposed on the three random variables $(\Omega,X,Y)$. When restricting attention to Blackwell information orders, our paper and \cite{he2026private} use the same notion for evaluating the informativeness of $(X,Y)$ about the state $\omega$. The key difference lies in the independence constraint. In \cite{he2026private}, private-private information requires $X$ and $Y$ to be independent, while secrecy in our model requires $Y$ and $\Omega$ to be independent. This difference is crucial, since the value of the sender's message in our model comes precisely from its dependence on $X$. Furthermore, our main dominance criterion is weaker than Blackwell dominance, since it compares information structures only under utilities in $\mathcal U$. Thus, relative to the Blackwell-Pareto frontier in \cite{he2026private}, our frontier problem differs first in the independence restriction and then in the weaker, preference-restricted criterion used to evaluate informativeness.

\subsection{Additional Discussions on Modeling Assumptions}\label{sec:discussion_assumption}
\paragraph{On the Possibility of Repeated Observations}
Our model assumes that, while the marginal distributions of the baseline message $x$ and the sender's message $y$ conditional on the state are publicly observable or verifiable, their joint distribution is shared only between the sender and the receiver. This assumption is innocuous in a one-shot interaction: the outside observer sees only a single realization and therefore cannot infer the joint distribution. When the sender and receiver interact repeatedly, however, the assumption requires an additional qualification. Repeated observations may be useful, or even necessary, for verifying the marginal distributions, for example to confirm that the sender's message is indeed secret. At the same time, if the same joint distribution were used repeatedly and both $x$ and $y$ were publicly observed, repeated observations could eventually reveal their dependence.

Thus, a repeated-interaction interpretation is most natural in environments in which the marginals are routine and publicly verifiable, but the dependence between the two messages is not fixed or publicly learnable across occasions. In such environments, $f$ and $g$ describe ordinary messages that are repeatedly generated and publicly observed. The state may correspond to sensitive information that the sender's public message is required not to reveal on its own, whether for legal, institutional, or strategic reasons. Each message therefore looks ordinary in isolation: the baseline message follows its usual distribution, and the sender's message follows its usual public distribution. What may vary across occasions is whether, and how, these two messages are correlated. The model is therefore particularly relevant in environments where, against this background of routine behavior, the sender and receiver can use the hidden dependence between $x$ and $y$ to transmit additional information without changing the marginal distribution of either message.

For example, a firm's CEO may regularly issue public updates, while a trader regularly receives market news from a particular source. Even if the CEO's public communication is constrained not to reveal the firm's future prospects on its own, its hidden correlation with the trader's news may still convey useful information. This provides a natural form of covert communication under realistic constraints: the channel operates through correlation rather than through changes in publicly observable behavior, and the question is precisely when such hidden correlation can remain both secret and plausibly deniable.

The preceding interpretation does not require covert communication to occur in every interaction between the sender and receiver. The sender and receiver may interact repeatedly, while using the hidden dependence between $x$ and $y$ only occasionally. Even if they do communicate in this way repeatedly, however, the model need not be interpreted as using the same joint structure in every round. The sender and receiver may update their codebook after each round, or after a small number of rounds, while preserving the publicly observable marginals. This does not require new private communication before each instance. Instead, they may agree in advance on a sequence of relabeling rules, so that subsequent communication takes place entirely through the public channel.\footnote{Arbitrary relabelings of named messages need not preserve the marginal distribution of $Y$, but the signal-representation formulation provides a convenient implementation. For a signal-based structure, $[0,1]$ can be partitioned into equal-measure cells, each corresponding to a realization of a refined sender's message. Relabeling these cells preserves the marginal distribution over the refined message space. By Theorem \ref{thm:frontier_secrecy_PD}, restricting attention to signal-based structures is without loss for the sender's frontier problem.} To an outside observer who does not know the relabeling rule, observations of $(x,y)$ are then generated by a sequence of independent but non-identically distributed joint information structures, rather than by repeated draws from a single fixed joint distribution. Thus, repeated observations can make the marginals verifiable without necessarily making the hidden dependence used in any particular instance identifiable.

\paragraph{On the Single-Crossing Assumption}
The single-crossing assumption is meant to capture the intuitive ordinal idea that higher states favor higher actions, while keeping the receiver's utility sufficiently flexible. This flexibility is central to the tractability of plausible deniability. Under single crossing, the rationalizable actions following a message can be read directly from the shape of its likelihood: decreasing likelihoods rationalize lower actions, increasing likelihoods rationalize higher actions, and non-monotone likelihoods are flexible enough to rationalize all actions. This yields the simple monotonicity-based characterization in Theorem \ref{thm:PD-characterization}. 

This tractability would generally be lost under stronger restrictions on utility. For example, with increasing differences, plausible deniability would be characterized indirectly through first-order stochastic dominance comparisons between the posterior and the prior, rather than directly through likelihood monotonicity. If the receiver's ordinal preference over state-action pairs were publicly known and only cardinal intensities were private, rationalizability would become even more action-specific: a message could rationalize an action only if it shifts likelihood toward states where that action is ordinally favored relative to the relevant alternatives. Thus, stronger utility restrictions would replace the simple monotonicity-based classification of messages with more detailed posterior- or action-specific comparisons.

In this sense, single crossing hits the right balance for our purposes. It imposes meaningful discipline on the receiver's behavior while preserving enough flexibility to obtain a tractable characterization of plausible deniability and the results developed above. We view the characterization of plausible deniability and optimal communication under stronger utility restrictions as an interesting question for future work.

\section{Conclusion}\label{sec:conclusion}

This paper studies communication when secrecy of the message does not guarantee secrecy of its use. A message may be statistically uninformative to outsiders, yet the action it induces may reveal that hidden information has been received. This creates a second design problem: the receiver's response must remain justifiable without invoking the message itself. We formalize this requirement as plausible deniability and show that, together with secrecy, it imposes a sharp geometry on feasible communication. At the frontier, informative refinements are disciplined by the receiver's baseline information; and, in the cases where a greatest feasible structure exists, the receiver learns interval information about the state. The result is a theory of how much information can be conveyed while keeping both the message and the receiver's response defensible to outsiders.

The same logic points to a positive role for secrecy and plausible deniability in protective information design. In reporting, compliance, whistleblowing, and related institutional settings, the goal is often to let sensitive information guide action without exposing the person or channel through which it was conveyed. A reporting protocol should therefore do more than hide the report itself: it should also allow any subsequent intervention to be justified on independent grounds. The framework developed here provides a way to study this tradeoff between informativeness and protection. A natural direction for future work is to apply the framework to specific institutional settings, where one can study how the design of reporting channels, interventions, and source protections shapes the feasible informativeness of communication.

\appendix
\part*{Appendix}
\numberwithin{equation}{section}
\numberwithin{lemma}{section}
\numberwithin{theorem}{section}
\numberwithin{proposition}{section}
\numberwithin{figure}{section}

\section{Proofs for Section \ref{sec:PD-secrecy}}

\subsection{Proof of Lemma \ref{lem:PD-R}}

\begin{proof}[Proof of Lemma \ref{lem:PD-R}]
    \textbf{If.} Fix $(x,y) \in \text{supp}(h) $ and $u \in \mathcal{U}$. By definition, $ a_{x,y}(u) \in \mathcal{R}_h(x,y)$. Then $ \mathcal{R}_{h}(x,y) \subseteq \mathcal{R}_f(x)$ implies that $a_{x,y}(u) \in \mathcal{R}_f(x)$. Therefore, there exists $\tilde{u} \in \mathcal{U}$ such that $a_{x,y}(u)$ is optimal under the posterior induced by $x$ and utility $\tilde{u}$. Thus, $h$ satisfies plausible deniability. 

    \medskip\noindent\textbf{Only if.} For any $(x,y) \in \text{supp}(h)$ and $a \in \mathcal{R}_h(x,y)$, there exists $u \in \mathcal{U}$ such that $a$ is optimal under the posterior induced by observing $(x,y) $ for the utility function $u$. By plausible deniability, there exists $\tilde{u} \in \mathcal{U}$ such that $a$ is also optimal under the posterior induced by observing $x$ for the utility function $\tilde{u}$, i.e., $a$ is rationalizable at $x$ under $f$ and $a \in \mathcal{R}_f(x)$. Thus, $\mathcal{R}_h(x,y) \subseteq \mathcal{R}_f(x)$. 
\end{proof}

\subsection{Proof of Lemma \ref{lem:rationalizability}}

    To prove Lemma \ref{lem:rationalizability}, we introduce the notion of single-crossing incremental returns.\footnote{See \cite{athey2018value}.} Formally, define
    \begin{equation*}
        R^{SC} := \left\{r\in \R^{\Omega}: \forall \omega' > \omega, r(\omega) >(\geq) 0 \implies r(\omega') >(\geq) 0 \right\}
    \end{equation*}
    as the set of all real-valued single-crossing functions. Each element $r\in R^{SC}$ is called a single-crossing incremental return. Hence, a utility function $u$ satisfies SCP if and only if $u(a', \omega) - u(a,\omega) \in R^{SC}$ for all $a' > a$. Next, define
    \begin{align*}
         R^-_\mu := \bigg\{r\in R^{SC}: \sum_{\omega}r(\omega)\mu(\omega) < 0\bigg\} \text{ and }  R^+_\mu := \bigg\{r\in R^{SC}: \sum_{\omega}r(\omega)\mu(\omega) > 0\bigg\}
    \end{align*}
    as the subsets of single-crossing incremental returns with negative and positive expectations under prior $\mu$, respectively. 
    The next lemma characterizes rationalizability using single-crossing incremental returns.
    \begin{lemma}\label{lem:incremental_return}
        Given an information structure $\phi:\Omega\to\Delta(Z)$, a message realization $z\in Z$ and an action $a > a_0$, $a \in \mathcal{R}_\phi(z)$ if and only if there exists $r\in  R^-_\mu$ such that
        \begin{equation}\label{ineq:incremental}
        \sum_{\omega} r(\omega)\phi(z|\omega)\mu(\omega) \geq 0.
        \end{equation}
        Symmetrically, for any action $a < a_0$, $a\in \mathcal{R}_\phi(z)$ if and only if there exists $r\in  R^+_\mu$ such that the reversed version of the above weak inequality holds.
    \end{lemma}

    \begin{proof}[Proof of Lemma \ref{lem:incremental_return}:]
    Without loss, we focus on the case when $a > a_0$.
    \medskip
    
    \noindent \textbf{Only if.}
    Since $a \in \mathcal{R}_\phi(z)$, there exists $u \in \mathcal{U}$ such that $a$ is optimal upon observing $z$. Define $r(\omega) = u(a, \omega) - u(a_0, \omega)$. Then $r\in  R^-_\mu$ and \eqref{ineq:incremental} follows.
    \medskip
    
    \noindent \textbf{If.}
    By $r \in R^-_\mu$ and \eqref{ineq:incremental}, there exists $\overline{\omega}$ such that $r(\omega) < 0$ for all $\omega < \overline{\omega}$ and $r(\omega) \geq 0$ for all $\omega \geq \overline{\omega}$. Define an utility function $u:A\times\Omega\to\mathbb{R}$ as
	\begin{equation*}
		u(\tilde{a}, \omega) := 
		\begin{cases}
			r(\omega) , & \text{if } \tilde{a} \ge a , \\
			r(\omega)/2 , & \text{if } a> \tilde{a} > a_0 , \\
			0 , & \text{if } \tilde{a} = a_{0} , \\
			\epsilon_1 , & \text{if } \tilde{a} < a_{0} ,\omega < \overline{\omega}\\
            -\epsilon_2 & \text{if } \tilde{a} < a_{0} ,\omega \geq \overline{\omega}
		\end{cases}
	\end{equation*}
    where $\epsilon_1, \epsilon_2$ are positive and satisfy
    \begin{equation}\label{ineq:restriction_epsilon}
        \frac{\epsilon_2}{\epsilon_1} > \max\left\{\frac{\sum_{\omega < \overline{\omega}}\mu(\omega)}{\sum_{\omega\geq\overline{\omega}}\mu(\omega)}, \frac{\sum_{\omega < \overline{\omega}}\phi(z|\omega)\mu(\omega)}{\sum_{\omega \geq \overline{\omega}}\phi(z|\omega)\mu(\omega)}\right\}.
    \end{equation}
    Note that $\sum_{\omega \geq \overline{\omega}}\mu(\omega)\neq 0$ because the prior has full support and $\sum_{\omega \geq \overline{\omega}}\phi(z|\omega)\mu(\omega) \neq 0$, otherwise \eqref{ineq:incremental} cannot hold. Thus, RHS of \eqref{ineq:restriction_epsilon} is finite. The choice of $\epsilon_1, \epsilon_2$, together with $r\in  R^-_\mu$ and \eqref{ineq:incremental}, implies $u$ satisfies SCP, $a_0$ is uniquely optimal under $\mu$ and $a$ is optimal conditional on $z$. Thus, $a \in \mathcal{R}_\phi(z)$.
    \end{proof}

    \begin{remark}
    The utility constructed in the proof is deliberately simple and may exhibit large sets of indifferent actions, and intermediate actions between $a$ and $a_0$ are optimal only when indifferent to $a$. This is purely an artifact of the convenient normalization used in the proof.

    If one prefers a more generic construction, one can proceed by fixing the same cutoff state $\bar\omega$ and constructing all adjacent incremental returns with this same cutoff. In this case, every pairwise difference is a positive sum of same-cutoff single-crossing functions and hence satisfies SCP. By choosing the magnitudes of these incremental returns appropriately, one can ensure that their expectations change sign at different thresholds of the tail probability $Pr(\{\omega\ge \bar\omega\})$. This yields a single-crossing utility under which all actions can be made uniquely optimal at suitable beliefs.
    \end{remark}

    \begin{proof}[Proof of Lemma \ref{lem:rationalizability}] Fix $z$. Denote $q(\omega) := \phi(z|\omega)$ and $\eta(\omega) := \frac{q(\omega)\mu(\omega)}{\sum_{\omega'}q(\omega')\mu(\omega')}$ as the posterior upon observing $z$.
     
    \noindent\textbf{(i) and (ii).} We prove (i), and (ii) follows symmetrically. Suppose $q(\omega)$ is decreasing and non-constant in $\omega$.

    Towards a contradiction,
    suppose $a \in \mathcal{R}_\phi(z)$ but $a >a_0$. By Lemma \ref{lem:incremental_return}, there exists $r \in  R^-_\mu$ such that \eqref{ineq:incremental} holds. Let $\overline{\omega} := \min \{ \omega : r(\omega) \geq 0 \}$ so that $r(\omega) < 0$ for all $\omega < \overline{\omega}$ and $r(\omega) \geq 0$ for all $\omega \geq \overline{\omega}$. Write $r = r^{+} - r^{-}$, where $r^{+}(\omega) := \max \{ r(\omega), 0 \}$  and $r^{-}(\omega) := \max \{ -r(\omega), 0 \}$. Note that 
    \begin{align*}
        \sum_{\omega}\mu(\omega)q(\omega)r(\omega) &= \sum_{\omega \geq \overline{\omega}}\mu(\omega)q(\omega)r^+(\omega) - \sum_{\omega < \overline{\omega}}\mu(\omega)q(\omega)r^-(\omega)\\
        &\leq q(\overline{\omega})\bigg[\sum_{\omega \geq \overline{\omega}}\mu(\omega)r^+(\omega) - \sum_{\omega < \overline{\omega}}\mu(\omega)r^-(\omega)\bigg]\\
        &= q(\overline{\omega})\sum_{\omega}\mu(\omega)r(\omega),
    \end{align*}
    where the inequality follows since $q$ is decreasing. Because $r \in R^-_\mu$, $\sum_{\omega} \mu(\omega)r(\omega) < 0$. If $q(\overline{\omega}) > 0$, then it follows that $\sum_{\omega} \mu(\omega)q(\omega)r(\omega) < 0$, which contradicts \eqref{ineq:incremental}. If $q(\overline{\omega}) = 0$, then $q(\omega) = 0$ for all $\omega \geq \overline{\omega}$, so that $\sum_{\omega} \mu(\omega)q(\omega)r(\omega) = - \sum_{\omega<\overline{\omega} } \mu(\omega)q(\omega)r^{-}(\omega) < 0.$ Again, it contradicts \eqref{ineq:incremental}. Thus, $\mathcal{R}_\phi(z)\subseteq\{a: a \le a_0\}$.

    Since $q$ is non-constant and decreasing, there exists $\hat{\omega}$ such that $q(\omega') > q(\hat{\omega}) \geq q(\omega)$ for all $\omega' < \hat{\omega}$ and all $\omega \geq \hat{\omega}$. Then, we have 
	\begin{equation*}
		\left ( \sum_{\omega < \hat{\omega}} \mu(\omega) \right )\cdot 
		\left ( \sum_{\omega \geq \hat{\omega}} q(\omega) \mu(\omega) \right ) \leq 
		\left ( \sum_{\omega < \hat{\omega}} \mu(\omega) \right )\cdot 
    		\left ( \sum_{\omega \geq \hat{\omega}} \mu(\omega)  \right ) q(\hat{\omega}) < 
		\left ( \sum_{\omega \geq \hat{\omega}} \mu(\omega) \right )\cdot 
		\left ( \sum_{\omega < \hat{\omega}} q(\omega) \mu(\omega) \right ).
    \end{equation*}	
    Using $\sum_{\omega < \hat{\omega}} \mu(\omega) +\sum_{\omega \ge \hat{\omega}} \mu(\omega) = 1$, we can rearrange the above inequality: 
    \begin{equation*}
	    \sum_{\omega\ge \hat{\omega}}q(\omega) \mu(\omega) < 
	    \left ( \sum_{\omega \ge \hat{\omega}} \mu(\omega) \right ) \cdot \left ( \sum_{\omega < \hat{\omega}}q(\omega) \mu(\omega) + \sum_{\omega \ge \hat{\omega}}q(\omega)\mu(\omega) \right ). 
    \end{equation*}
    So we have $\sum_{\omega \geq \hat{\omega}} \eta(\omega) < \sum_{\omega \geq \hat{\omega}} \mu(\omega)$. Pick $c > 0$ such that $\sum_{\omega \geq \hat{\omega}} \eta(\omega) < c < \sum_{\omega \geq \hat{\omega}} \mu(\omega)$, and define $r(\omega) := \mathbf{1} \{ \omega \geq \hat{\omega} \} - c$. Thus, $r \in R^{SC}$. Moreover, it holds that $\sum_{\omega} r(\omega)\mu(\omega) = \sum_{\omega \geq \hat{\omega}} \mu(\omega)  - c > 0$,
    thus, $r \in R^+_\mu$. On the other hand, $\sum_{\omega} r(\omega)q(\omega)\mu(\omega) = \left(\sum_{\omega \geq \hat{\omega}}\eta(\omega) - c \right) \sum_{\omega} \mu(\omega)q(\omega) < 0.$ By Lemma \ref{lem:incremental_return}, any $a < a_0$ is rationalizable at $z$. The case when $a = a_{0}$ is straightforward. Hence, $\{a: a \le a_0\}\subseteq\mathcal{R}_\phi(z)$. This proves (i).

    \medskip
    \noindent\textbf{(iii).} If $q(\omega)$ is constant in $\omega$, the posterior is the same as the prior. For all $u \in \mathcal{U}$, $a_0$ is uniquely optimal under the prior. Thus, $\mathcal{R}_\phi(z) = \{a_0\}$.

    \medskip
    \noindent\textbf{(iv).} Suppose $q(\omega)$ is non-monotone in $\omega$. We show that $A\subseteq \mathcal{R}_\phi(z)$. Consider any action $a > a_{0}$. By Lemma \ref{lem:incremental_return}, it suffices to construct a $r\in R^-_\mu$ such that \eqref{ineq:incremental} holds. Since $q$ is not decreasing, there exists $i < j$ such that $q(\omega_{i}) < q(\omega_{j})$. Let $P := \{\omega: \omega < \omega_i\}$ and $Q := \Omega\setminus( P\cup\{\omega_i, \omega_j\})$.\footnote{$P$ and $Q$ can be empty and the proof still applies.} Define
    \begin{align*}
	r(\omega) := 
	\begin{cases}
		- \epsilon, & \text{if } \omega \in P,\\
		-1, & \text{if } \omega = \omega_{i},\\
		M, & \text{if } \omega = \omega_{j},\\
		\epsilon, & \text{if } \omega \in Q.
	\end{cases}
    \end{align*}
    where $\epsilon, M > 0$. Then, $r \in R^{SC}$, and we have 
    \begin{align*}
	    & \sum_{\omega} \mu(\omega)r(\omega) =  \epsilon [\mu(Q) - \mu(P)] - \mu(\omega_{i}) + M \mu(\omega_{j}),\footnotemark \\
	    & \sum_{\omega} \mu(\omega)q(\omega)r(\omega) =  \epsilon \left(\sum_{\omega\in Q}\mu(\omega)q(\omega) - \sum_{\omega \in P}\mu(\omega)q(\omega)\right) - \mu(\omega_{i})q(\omega_{i}) + M \mu(\omega_{j})q(\omega_{j}).
    \end{align*}
    \footnotetext{For any $A \subseteq \Omega$, $\mu(A) := \sum_{\omega \in A} \mu(\omega)$. }
    Choose $M$ satisfying 
    \begin{align*}
	    \frac{\mu(\omega_{i})q(\omega_{i}) - \epsilon \left(\sum_{\omega\in Q}\mu(\omega)q(\omega) - \sum_{\omega \in P}\mu(\omega)q(\omega)\right)}{\mu(\omega_{j})q(\omega_{j})} < M < \frac{\mu(\omega_{i}) - \epsilon [\mu(Q) - \mu(P)]}{\mu(\omega_{j})}.
    \end{align*}
    Such an $M$ exists for all sufficiently small $\epsilon > 0$, because as $\epsilon \to 0$, the left and right hand sides converge to 
    \begin{align*}
	    \frac{\mu(\omega_{i})q(\omega_{i})}{\mu(\omega_{j})q(\omega_{j})} < \frac{\mu(\omega_{i})}{\mu(\omega_{j})}.
    \end{align*}
    With the above choice of $M$, it holds that
    \begin{align*}
	    \sum_{\omega} \mu(\omega)r(\omega) < 0, \quad \text{and} \quad \sum_{\omega} \mu(\omega)q(\omega)r(\omega) > 0.
    \end{align*}
    Thus, $r \in R^-_\mu$ and $\sum_{\omega} \mu(\omega)q(\omega)r(\omega) > 0$. A symmetric argument applies to any $a \le a_0$, hence $\mathcal{R}_\phi(z) = A$.
\end{proof}

\subsection{Proof of Proposition \ref{prop:PD-greatest}}

\begin{proof}[Proof of Proposition~\ref{prop:PD-greatest}]
We argue the case $d \in D$; the cases $i \in I$ and $s \in S$ are symmetric, replacing $C^D$ by the cone of nonnegative increasing vectors (extreme rays $\mathbf{e}_{\geq k}$) and by $\mathbb{R}_+^n$ (extreme rays $\mathbf{e}_k$), respectively.

For any joint structure $\phi$ and signal $z$, write $\phi^{z} := [\phi(z| \omega_1), \ldots, \phi(z| \omega_n)]^{\intercal}$. By Theorem~\ref{thm:PD-characterization}, every $h^{d,y}$ with $h \in \mathcal{H}^{PD}$ lies in the cone of nonnegative decreasing vectors
\begin{equation*}
C^D := \{v \in \mathbb{R}_+^n : v_1 \geq v_2 \geq \cdots \geq v_n\},
\end{equation*}
whose extreme rays are spanned by $\mathbf{e}_{\leq 1}, \ldots, \mathbf{e}_{\leq n}$, where $(\mathbf{e}_{\leq k})_i := \mathbf{1}\{i \leq k\}$. Each $v \in C^D$ has the unique telescoping decomposition
\begin{equation}\label{eq:telescoping}
v = \sum_{k=1}^n (v_k - v_{k+1})\,\mathbf{e}_{\leq k}, \qquad v_{n+1} := 0.
\end{equation}
Applied to $f^d$, this yields $f^d = \sum_{k=1}^n c_k\,\mathbf{e}_{\leq k}$ with $c_k := f(d | \omega_k) - f(d | \omega_{k+1})$.

\medskip

\noindent\textbf{If.} Let $X \times \overline{Y}$ be the signal space of $\overline{h}$. By hypothesis, $\overline{h}^{d,y} = \beta(y)\,\mathbf{e}_{\leq k(y)}$ for some $\beta(y) \geq 0$ and $k(y) \in \{1,\ldots,n\}$; setting $\mathcal{Y}_k := \{y \in \overline{Y} : k(y) = k\}$, the marginal condition and uniqueness of \eqref{eq:telescoping} force $\sum_{y \in \mathcal{Y}_k} \beta(y) = c_k$, for every $k$. 

Now fix $h \in \mathcal{H}^{PD}$ with signal space $X \times Y$. Applying \eqref{eq:telescoping} columnwise,
\[
h^{d,y'} = \sum_{k=1}^n \alpha_k(y')\,\mathbf{e}_{\leq k}, \qquad \alpha_k(y') := h(d,y' | \omega_k) - h(d,y' | \omega_{k+1}) \geq 0.
\]
Summing over $y'$ and using $\sum_{y' \in Y} h^{d,y'} = f^d$ with uniqueness of \eqref{eq:telescoping}  gives $\sum_{y' \in Y} \alpha_k(y') = c_k$ for each $k$; in particular, $\alpha_k \equiv 0$ whenever $c_k = 0$. Define $\gamma(d, y' | d, y) := \alpha_{k(y)}(y')/c_{k(y)}$ when $\beta(y) > 0$, with $\gamma(\tilde{x}, y' | d, y) = 0$ for $\tilde{x} \neq d$ and $\gamma(\cdot | d, y)$ arbitrary on $\{d\} \times Y$ when $\beta(y) = 0$. Each $\gamma(\cdot | d, y)$ is a probability measure, and
\begin{align*}
\sum_{y \in \overline{Y}} \overline{h}(d, y | \omega)\,\gamma(d, y' | d, y)
= \sum_{k : c_k > 0} \Bigg( \sum_{y \in \mathcal{Y}_k} \beta(y) \Bigg) \frac{\alpha_k(y')}{c_k}\,\mathbf{e}_{\leq k}(\omega)
= h(d, y' | \omega).
\end{align*}
As the same argument applies to all three cases, $h$ can be obtained from $\overline{h}$ by garbling.

\medskip
\noindent\textbf{Only if.} Let $h \in \mathcal{H}^{PD}$ be arbitrary. If $h$ does not satisfy the conditions, then there exists $(x,y) \in \text{supp}(h)$ such that $h^{x,y}$ is not an extreme ray. Then it is feasible to replace $y$ by $y_{1}$ and $y_{2}$ such that $h^{x, y_{1}}$ and $h^{x, y_{2}}$ are not aligned and $h^{x, y_{1}} + h^{x, y_{2}} = h^{x,y}$. Leaving the other signals unchanged, this refinement is strictly Blackwell more informative than the original one. 
\end{proof}

\subsection{Proof of Proposition \ref{prop:secrecy-frontier}}
\begin{proof}[Proof of Proposition \ref{prop:secrecy-frontier}]
Let $X=\{x_1,\dots,x_k\}$ and $Y=\{y_1,\dots,y_m\}$. Consider any joint information structure $h:\Omega\to\Delta(X\times Y)$ that satisfies secrecy. By secrecy, the $Y$-marginal of $h$ is independent of $\omega$. Define $g(y_j):=\sum_{i=1}^k h(x_i,y_j| \omega)$, which is well-defined. Let $G(y_j):=\sum_{j'=1}^j g(y_{j'})$ for $j=1,\dots,m$, with $G(y_0):=0$. Thus $\{G(y_j)\}_{j=0}^m$ partitions $T=[0,1]$.

We construct a signal representation $\psi:\Omega\times T\to X$ as follows. Fix
$\omega\in\Omega$. For each $j$ and each $t\in(G(y_{j-1}),G(y_j)]$, define
$\psi(\omega,t)=x_i$ if
\[
G(y_{j-1})+\sum_{i'=1}^{i-1} h(x_{i'},y_j| \omega)
<
t
\le
G(y_{j-1})+\sum_{i'=1}^{i} h(x_{i'},y_j| \omega).
\]
That is, within each interval corresponding to $y_j$, we allocate mass
according to the conditional distribution of $x$ given $(y_j,\omega)$.

We next show that $\xi_\psi$ Blackwell dominates $h$.
Consider the garbling that maps $(x,t)$ to $(x,y_j)$ whenever $G(y_{j-1})<t\le G(y_j)$. Then for any $\omega\in\Omega$, $x_i\in X$, and $y_j\in Y$,
\begin{align*}
\int_{G(y_{j-1})}^{G(y_j)} \xi_{\psi} (x_i,t| \omega)\,dt =\int_{G(y_{j-1})+\sum_{i'=1}^{i-1} h(x_{i'},y_j| \omega)}^{G(y_{j-1})+\sum_{i'=1}^{i} h(x_{i'},y_j| \omega)} 1\,dt  = h(x_i,y_j| \omega).
\end{align*}
Thus, $h$ can be obtained from $\xi_{\psi}$ via this garbling, implying that $\xi_{\psi}$ is Blackwell more informative than $h$. By letting $\tilde{h}$ to be the induced joint structure $h_\psi$, the result follows.
\end{proof}

\section{Proofs for Section \ref{sec:frontier_optimal}}
\subsection{Proof of Theorem \ref{thm:frontier_secrecy_PD}}

\begin{proof}[Proof of Theorem \ref{thm:frontier_secrecy_PD}]
Fix any $h\in\mathcal H^{SPD}$. We follow the construction in the proof of
Proposition \ref{prop:secrecy-frontier}. Since $h$ satisfies secrecy, define
$g(y_j):=\sum_{i=1}^k h(x_i,y_j\mid\omega)$ and
$G(y_j):=\sum_{j'=1}^j g(y_{j'})$, so that $\{G(y_j)\}_{j=0}^m$ partitions
$T=[0,1]$ into blocks indexed by $Y$.

For each $y\in Y$, decompose the vectors $h^{x,y}$ as in the proof of Proposition \ref{prop:PD-greatest}: for $d\in D$ into extreme rays of $C^D$, for $i\in I$ into extreme rays of $C^I$, and for $s\in S$ into nonnegative multiples of the coordinate rays. Then, within the block corresponding to $y$, refine the segment assigned to $(x,y)$ according to this decomposition. Each refined subsegment therefore induces exactly one extreme ray: a decreasing one for $d\in D$, an increasing one for $i\in I$, and an unrestricted one for $s\in S$. This yields a signal representation $\psi:\Omega\times T\to X$ and hence an induced density $\xi_\psi$ over $X\times T$.

Exactly as in Proposition \ref{prop:secrecy-frontier}, the garbling that maps $(x,t)$ to $(x,y_j)$ whenever $G(y_{j-1})<t\le G(y_j)$ recovers $h$, so $\xi_\psi$ weakly Blackwell dominates $h$. By construction, every realized refined signal induces a likelihood vector satisfying the monotonicity restriction in Theorem \ref{thm:PD-characterization}. The induced joint structure $h_\psi$ inherits this property and is therefore an element in $\mathcal H^{SPD}$. The result follows by letting $\tilde{h} = h_\psi$.
\end{proof}

\subsection{Proof of Theorem \ref{thm:directional_almost_directional}}

\begin{lemma}\label{lem:swapping}
    Suppose the baseline information structure $f$ is almost-directional and $\psi$ is a representation of $f$ such that $h_\psi$ is a frontier point of $\mathcal H^{SPD}$. Then, for any $\omega'' > \omega'$, the following conditions cannot hold simultaneously:
    \begin{enumerate}
        \item there exists a non-measure zero set $T_0\subseteq T$ such that $\psi(\omega', t)\in D$ and $\psi(\omega'', t)\in I$ for all $t\in T_0$;
        \item there exists a non-measure zero set $T_1\subseteq T$ such that $\psi(\omega', t)\in S$ and $\psi(\omega'', t)\in S$ for all $t\in T_1$.
    \end{enumerate}
\end{lemma}

\begin{proof}[Proof of Lemma \ref{lem:swapping}] If $S = \emptyset$, then the lemma holds trivially. Consider $S = \{s\}$. Suppose there exists $T_0, T_1 \subseteq T$ with equal measure, $\omega'' > \omega'$, $d\in D$ and $i\in I$ such that\footnote{It is without loss to assume that $T_0$ and $T_1$ have equal measure and that $\psi(\omega,\cdot)$ is constant on each of these sets for every $\omega$, since we can always shrink $T_0$ and $T_1$ to subsets satisfying these properties.}
\begin{align*}
    &\psi(\omega', t) = d \quad\text{and}\quad \psi(\omega'', t) = i  \quad\text{for all }t\in T_0;\\
    &\psi(\omega', t) = s \quad\text{and}\quad \psi(\omega'', t) = s \quad\text{for all }t\in T_1.
\end{align*}
Without loss, let $\omega' = \max\{\omega: \psi(\omega, t) = d, t\in T_0\}, \omega'' = \min\{\omega: \psi(\omega, t) = i, t\in T_0\}$. The proof proceeds by considering two cases. In each case, we can find an alternative representation $\tilde{\psi}$ such that $h_{\tilde{\psi}}$ dominates $h_\psi$ which contradicts to $h_\psi$ is a frontier point. With slight abuse of notation, we use $T_0$ and $T_1$ also to denote the corresponding sender-message labels in the finite joint information structures induced by $\psi$ and $\tilde\psi$, after refining the partition of $T$ so that $T_0$ and $T_1$ are separate cells.

\medskip
\noindent\textbf{Case 1.} There is no $\omega$ between $\omega'$ and $ \omega''$. 

\noindent Consider the following subcases:
\begin{enumerate}[(i)]
    \item For all $t\in T_1$, $\psi(\omega, t) = s$ for all $\omega \geq \omega''$.
    \item For all $t\in T_1$, $\psi(\omega, t) = s$ for all $\omega \leq \omega'$.
    \item There exists  $\overline{\omega} > \omega''$ and $i'\in I$ such that $\psi(\omega, t) = i'$ for all $t\in T_1, \omega \geq \overline{\omega}$.
\end{enumerate}
In subcase (i), consider an alternative representation $\tilde{\psi}$ defined by
\begin{align*}
    \tilde{\psi}(\omega, t) = \begin{cases}
        s &\text{for all } t\in T_0, \omega\geq \omega''\\
        i &\text{for all } t\in T_1, \omega\geq\omega''\\
        \psi(\omega, t) &\text{otherwise}.
    \end{cases} 
\end{align*}
$h_\psi$ can be obtained from $h_{\tilde{\psi}}$  by a garbling that sends $(s, T_0)$ to $(i, T_0)$, $(i, T_1)$ to $(s, T_1)$ and identity elsewhere. Thus, $h_{\tilde{\psi}}$ strictly Blackwell dominates $h_\psi$. 
    
Similarly, in subcase (ii), consider another representation $\tilde{\psi}$ defined by
\begin{align*}
    \tilde{\psi}(\omega, t) = \begin{cases}
        s &\text{for all } t\in T_0, \omega\leq \omega'\\
        d &\text{for all } t\in T_1, \omega\leq\omega'\\
        \psi(\omega, t) &\text{otherwise}.
    \end{cases} 
\end{align*} 
$h_\psi$ can be obtained from $h_{\tilde{\psi}}$ by a garbling that sends $(s, T_0)$ to $(d, T_0)$, $(d, T_1)$ to $(s, T_1)$ and identity elsewhere. Thus, $h_{\tilde{\psi}}$ strictly Blackwell dominates $h_\psi$.

In subcase (iii), let $\overline{\omega}$ be the smallest state such that the condition holds. Consider an alternative representation $\tilde{\psi}$ defined by
\begin{align*}
    \tilde{\psi}(\omega, t) = 
    \begin{cases}
       s &\text{for all } t\in T_0, \overline{\omega} > \omega\geq \omega''\\
       i' &\text{for all } t\in T_0, \omega\geq \overline{\omega}\\
       i  &\text{for all } t\in T_1, \omega\geq \omega''\\
       \psi(\omega, t) &\text{otherwise}
    \end{cases}
\end{align*}
$h_\psi$ can be obtained from $h_{\tilde{\psi}}$ by a garbling that sends $(s, T_0)$ to $(s, T_1)$, $(i, T_1)$ to $(i, T_0)$ and $(i', T_0)$ to $(i', T_1)$ and identity elsewhere. So, $h_{\tilde{\psi}}$ strictly Blackwell dominates $h_\psi$. 

\medskip

\noindent\textbf{Case 2.} There exists some $\omega$ between $\omega'$ and $\omega''$.

\noindent It must be that $\psi(\omega,t) = s$ for all $\omega' < \omega < \omega''$ and $t\in T_0\cup T_1$. Let
\begin{align*}
    \overline{\omega} &= \max\{\omega\geq \omega'': \psi(\omega, t) = s, t\in T_1\}\\
    \underline{\omega} &= \min\{\omega\leq \omega': \psi(\omega, t) = s, t\in T_1\}.
\end{align*}
Therefore, under $h_\psi$, the posterior conditional on $(s, T_0)$ is the truncated prior over $\omega' < \omega < \omega''$ and the posterior conditional on $(s, T_1)$ is the truncated prior over $\underline{\omega} \leq \omega \leq \overline{\omega}$.

Consider an alternative representation $\tilde{\psi}$ defined by
\begin{align*}
    \tilde{\psi}(\omega, t) = 
    \begin{cases}
        s &\text{for all } t\in T_0, \omega'' \leq \omega \leq \overline{\omega}\\
        \psi(\omega, \nu(t)) &\text{for all } t\in T_0, \omega > \overline{\omega}\\
        i &\text{for all } t\in T_1, \omega\geq \omega''\\
        \psi(\omega, t) &\text{otherwise}
    \end{cases}
\end{align*}
where $\nu:T_0 \to T_1$ is a strictly increasing bijection from $T_0$ to $T_1$. Under $h_{\tilde{\psi}}$, the posterior conditional on $(s, T_0)$ is the truncated prior over $\omega' < \omega \leq \overline{\omega}$, while the posterior conditional on $(s, T_1)$ is the truncated prior over $\underline{\omega} \leq \omega < \omega''$. Moreover, the set of posteriors induced by signals other than $(s, T_0)$ and $(s, T_1)$ remains the same.

Let the optimal actions conditional on $(s, T_0)$ and $(s, T_1)$ under $h_\psi$ be $a$ and $a'$ respectively. Consider the following feasible action rule under $h_{\tilde{\psi}}$:
\begin{enumerate}
    \item chooses $\max\{a, a'\}$ conditional on $(s, T_0)$ and $\min\{a, a'\}$ conditional on $(s, T_1)$;
    \item choose the same actions as in the optimal action rule under $h_\psi$ when holding other posterior beliefs.
\end{enumerate}
We show that this action rule under $h_{\tilde{\psi}}$ gives higher expected payoff to a decision maker with single-crossing utility than the optimal action rule under $h_\psi$. 

Suppose $a > a'$. By the optimality of $a$ conditional on $(s, T_0)$, we have
\begin{equation*}
    \sum_{\omega' < \omega < \omega''} \mu(\omega)u(a, \omega) \geq \sum_{\omega' < \omega < \omega''} \mu(\omega)u(a', \omega).
\end{equation*}
This implies the existence of $\hat{\omega}$ between $\omega', \omega''$ such that $u(a, \hat{\omega}) \geq u(a', \hat{\omega})$. By SCP, $u(a, \omega) \geq u(a', \omega)$ for all $\omega > \hat{\omega}$. Now, note that the utility difference between the optimal action rule under $h_\psi$ and the action rule under $h_{\tilde{\psi}}$ defined above is proportional to 
\begin{equation*}
    \sum_{\omega' < \omega < \omega''} \mu(\omega)u(a, \omega) + \sum_{\underline{\omega} \leq \omega \leq \overline{\omega}}\mu(\omega)u(a', \omega) - \sum_{\omega' \leq \omega \leq \overline{\omega}}\mu(\omega)u(a, \omega) - \sum_{\underline{\omega} \leq \omega < \omega''}\mu(\omega)u(a', \omega).
\end{equation*}
After simplifying, we get 
    $-\sum_{\omega'' \leq \omega\leq \overline{\omega}}\mu(\omega)[u(a, \omega) - u(a'\omega)] \leq 0$
where the inequality follows from the fact that $u(a, \omega) - u(a', \omega)\geq 0$ for all $\omega\geq \omega'' > \hat{\omega}$.

If $a' > a$, the utility difference between the optimal action rule under $h_\psi$ and the action rule under $h_{\tilde{\psi}}$ defined above is proportional to 
\begin{equation*}
    \sum_{\omega' < \omega < \omega''} \mu(\omega)u(a, \omega) + \sum_{\underline{\omega} \leq \omega \leq \overline{\omega}}\mu(\omega)u(a', \omega) - \sum_{\omega' \leq \omega \leq \overline{\omega}}\mu(\omega)u(a', \omega) - \sum_{\underline{\omega} \leq \omega < \omega''}\mu(\omega)u(a, \omega).
\end{equation*}
After simplifying, we get 
$\sum_{\underline{\omega} \leq \omega \leq \omega'}\mu(\omega)[u(a', \omega) - u(a,\omega)]$.
If the difference is positive, then there must be some $\hat{\omega}$ between $\underline{\omega}, \omega'$ such that $u(a', \hat{\omega}) - u(a, \hat{\omega}) > 0$. But then SCP implies that $u(a',\omega) > u(a, \omega)$ for all $\omega' < \omega < \omega''$. This violates the optimality of $a$ on this interval. Thus, the utility difference is nonpositive for any $u\in \mathcal{U}$ and negative for some $u\in\mathcal{U}$. Therefore, $h_{\tilde{\psi}}$ dominates $h_{\psi}$.
\end{proof}

\medskip

\begin{proof}[Proof of Theorem \ref{thm:directional_almost_directional}]
Consider an arbitrary representation $\psi$ of $f$ such that the $D$-regions are left-aligned across states and $h_\psi$ satisfies the monotonicity requirements from plausible deniability. In particular,
$\psi(\omega, t) \in D$ for all $t < T_1(\omega) := \sum_{d\in D}f(d|\omega)$ for all $\omega$.

We prove by induction. First consider $\omega_1$. Suppose there exists intervals $T', T'' \subseteq [0, 1]$ with equal measure such that $\sup T' \leq \inf T''$, $\psi(\omega_1, t) \in I$ for all $t \in T'$ and $\psi(\omega_1, t) = s$ for all $t\in T''$. Define
\begin{align*}
    \tilde{\psi}(\omega, t) = 
    \begin{cases}
        \psi(\omega, \nu(t)) &\text{if } t\in T'\\
        \psi(\omega, \nu^{-1}(t)) &\text{if } t\in T''\\
        \psi(\omega, t) &\text{otherwise}
    \end{cases}
\end{align*}
for all $\omega$, where $\nu:T'\to T''$ is a strictly increasing bijection from $T'$ to $T''$. In words, $\tilde{\psi}$ is obtained by swapping the intervals $T', T''$ for all $\omega$. The resulting $h_{\tilde{\psi}}$ is equivalent to $h_\psi$ and plausible deniability is not affected. Repeat this swapping procedure until the signal representation is direction-ordered at $\omega_1$.  

Next, fix an arbitrary $k > 1$ and suppose $\psi$ is direction-ordered at all $\omega < \omega_k$. By the assumption that $D$-regions are left-aligned and plausible deniability, 
\begin{align*}
    \psi(\omega_k, t) \in \begin{cases}
        D &\text{if } t < T_1(\omega_k) := \sum_{d\in D}f(d|\omega_k)\\
        I &\text{if } t > T_2(\omega_{k - 1}) := 1 - \sum_{i\in I}f(i|\omega_{k - 1}).
    \end{cases}
\end{align*}
Suppose there exists intervals $T', T'' \subseteq [0, 1]$ with equal measure such that $\sup T' \leq \inf T''$, $\psi(\omega_k, t) = i\in I$ for all $t \in T'$ and $\psi(\omega_k, t) = s $ for all $t \in T''$.\footnote{As in Lemma \ref{lem:swapping}, we assume $\psi(\omega,\cdot)$ is constant on $T'$ and $T''$ for all $\omega$.} 

\medskip
\noindent\textbf{Case 1.} $T', T'' \subseteq [T_1(\omega_k), T_1(\omega_{k - 1})]$ where $T_1(\omega_{k - 1}) := \sum_{d\in D}f(d|\omega_{k - 1})$.\\
In this case, we can swap the two intervals for all states as above to construct an equivalent joint information structure while preserving plausible deniability.

\medskip
\noindent\textbf{Case 2.} $T' \subseteq [T_1(\omega_k), T_1(\omega_{k - 1})]$ and $T'' \subseteq [T_1(\omega_{k - 1}), T_2(\omega_{k - 1})]$.\\
In this case, we have $\psi(\omega_{k -1}, t) \in D$ and $\psi(\omega_k, t) \in I$ for all $t\in T'$ while $\psi(\omega_{k - 1}, t) = \psi(\omega_k, t) = s$ for all $t \in T''$. By  Lemma \ref{lem:swapping}, one can obtain an alternative representation $\tilde{\psi}$ by swapping the two intervals for all $\omega \geq \omega_k$ and $h_{\tilde{\psi}}$  Blackwell dominates $h_\psi$.

\medskip
\noindent\textbf{Case 3.} $T', T'' \subseteq [T_1(\omega_{k - 1}), T_2(\omega_{k - 1})]$.\\
In this case, for all $t\in T'\cup T''$, $\psi(\omega_{k - 1}, t) = s$ and $\psi(\omega, t)\in D\cup S$ for all $\omega < \omega_{k - 1}$.

Suppose there exists some $\omega' < \omega_{k - 1}$ such that $\psi(\omega', t)\in D$ for all $t\in T'$. If $\psi(\omega', t) \in D$ for all $t \in T''$, then we have $\psi(\omega, t) \in D$ for all $\omega \leq \omega'$ and $t\in T'\cup T''$. We can swap the two intervals for all states to obtain an equivalent joint structure as in case 1 above. If not, we must have $\psi(\omega', t) = s$ for all $t \in T''$. We can then improve $h_\psi$ as in case 2 of the proof of Lemma \ref{lem:swapping}. If there does not exist such $\omega'$, then it must be that $\psi(\omega, t) = s$ for all $\omega < \omega_k$ and $t \in T'\cup T''$. Again, we can swap the two intervals for all states to obtain an equivalent joint structure.

\medskip After each swap, the resulting joint structure weakly dominates the preceding one. Moreover, the swap moves the $I$-region originally on $T'$ to the right of the $S$-region originally on $T''$ at $\omega_k$, while preserving direction-ordering for all states $\omega<\omega_k$. Repeating this procedure finitely many times eliminates all such misorderings at $\omega_k$. Proceeding inductively over $k=2,\ldots,n$, we obtain a direction-ordered joint information structure that dominates the original one.

We next show that all direction-ordered structures satisfying plausible deniability are Blackwell equivalent. Fix two direction-ordered representations $\psi$ and $\psi'$ of $f$ that satisfy the monotonicity requirements of plausible deniability. Since $f$ is almost-directional, there is at most one message in $S$; hence, once the representation is direction-ordered, the $S$-region is pinned down between the $D$- and $I$-regions. For each $d\in D$, the likelihood vector $f(d|\cdot)$ has a unique decomposition into decreasing extreme rays. Therefore any direction-ordered representation satisfying plausible deniability induces, up to relabeling of sender messages, the same collection of lower-tail likelihood vectors associated with $d$. The same argument applies symmetrically to each $i\in I$. Hence all such direction-ordered structures are Blackwell equivalent.

Now fix any direction-ordered representation $\psi^*$ satisfying plausible deniability. Since every $h\in\mathcal H^{SPD}$ is dominated by some direction-ordered structure, and all such direction-ordered structures are Blackwell equivalent, $h_{\psi^*}$ dominates every element of $\mathcal H^{SPD}$. Therefore $h_{\psi^*}$ is a greatest element of $\mathcal H^{SPD}$.
\end{proof}

\subsection{Proof of Theorem \ref{thm:PD_greatest_if}}
\begin{proof}[Proof of Theorem \ref{thm:PD_greatest_if}]
Consider a direction-ordered representation $\psi$ of $f$. Arrange the messages within each of the monotone regions such that for all $\omega'' > \omega', d\in D$ and $i\in I$, $\psi(\omega'', t) = d$ implies $\psi(\omega', t) = d$ and $\psi(\omega', t) = i$ implies $\psi(\omega'', t) = i$. This ensures the induced $h_\psi$ satisfies the monotonicity requirement of plausible deniability.

Next, for any $d\in D$ and $y\in Y$, $h_\psi(d,y|\omega)=\int_{T_y}\mathbf 1\{\psi(\omega,t)=d\}dt$ is either $\lambda(T_y)$ or $0$ at each state $\omega$. Together with the monotonicity just imposed, this implies that $h_\psi^{d,y}$ lies on an extreme ray of the decreasing cone $C^D$. The argument for each $i\in I$ is symmetric.

It remains to verify that the non-monotone regions can be separated across states.
Since $\psi$ is direction-ordered, for each $\omega_k$ the $S$-region is the interval
\[
\left[
\sum_{d\in D} f(d|\omega_k),
\;
\sum_{d\in D} f(d|\omega_k)+\sum_{s\in S} f(s|\omega_k)
\right].
\]
We claim that these intervals are ordered from right to left across states, i.e.,\footnote{If
$\{t:\psi(\omega_k,t)\in S\}=\emptyset$, define both the infimum and the supremum as
$\sum_{d\in D} f(d|\omega_k)$.}
\[
\inf\{t:\psi(\omega_k,t)\in S\}
\geq
\sup\{t:\psi(\omega_{k+1},t)\in S\},
\qquad k=1,\ldots,n-1.
\]
For $k=1,\ldots,n-2$, this follows from the condition in Theorem
\ref{thm:PD_greatest_if} applied to the interior state $\omega_{k+1}$:
\[
\sum_{s\in S} f(s|\omega_{k+1})
\leq
\sum_{d\in D}\bigl[f(d|\omega_k)-f(d|\omega_{k+1})\bigr] \Leftrightarrow \sum_{d\in D} f(d|\omega_k)
\geq
\sum_{d\in D} f(d|\omega_{k+1})+\sum_{s\in S} f(s|\omega_{k+1}).
\]
For the boundary case $k=n-1$, the desired inequality is
\[
\sum_{d\in D} f(d|\omega_{n-1})
\geq
\sum_{d\in D} f(d|\omega_n)+\sum_{s\in S} f(s|\omega_n).
\]
Since probabilities across baseline messages add up to one in every state, the boundary inequality is equivalent to
\[
\sum_{s\in S} f(s|\omega_{n-1})
\leq
\sum_{i\in I}\bigl[f(i|\omega_n)-f(i|\omega_{n-1})\bigr],
\]
which follows from the condition in Theorem \ref{thm:PD_greatest_if} applied to
$\omega_{n-1}$.

Therefore the intervals $\left\{\{t:\psi(\omega,t)\in S\}\right\}_{\omega\in\Omega}$ are pairwise disjoint. Hence, for every $s\in S$ and $y\in Y$, $h_\psi(s,y|\omega)$ is either zero in every state or positive in exactly one state. By Proposition
\ref{prop:PD-greatest}, $h_\psi$ is the greatest element of $\mathcal H^{PD}$.
Since $h_\psi$ also satisfies secrecy, and $\mathcal H^{SPD}\subseteq\mathcal H^{PD}$,
it follows that $h_\psi$ is the greatest element of $\mathcal H^{SPD}$.
\end{proof}

\section{Additional Results on the Secrecy Frontier}
\label{app:secrecy-frontier-additional}

Theorem 1 of \cite{cavounidis2025} (Proposition \ref{prop:secrecy-fully-revealing}) shows that a fully revealing joint information structure satisfying secrecy exists if and only if $\sum_{\omega\in\Omega} f(x|\omega)\leq 1$ for every $x \in X$. When this condition fails for some message $x$, full revelation is infeasible because the total length required to assign message $x$ across all states exceeds the unit interval. A natural follow-up question is whether, whenever the condition holds for a particular message $x$, frontier structures must place the corresponding $x$-regions without overlap across states, so that observing $x$ together with the sender's message reveals the state. The answer is yes in the binary-state case, but not in general.

\begin{proposition}
\label{prop:secrecy_greatest}
Suppose $|\Omega|=2$. Let $x^* \in \arg\max_{x\in X}\bigl[f(x|\omega_1)+f(x|\omega_2)\bigr]$. Consider a signal representation $\psi^*:\Omega\times T\to X$ such that
\begin{align*}
\psi^*(\omega_1,t)=x^* \text{ for }t\in [0,f(x^*|\omega_1)), \text{ and } \psi^*(\omega_2,t)=x^* \text{ for }t\in [1-f(x^*|\omega_2),1], 
\end{align*}
with the remaining regions assigned so that $\psi^*$ is a signal representation of $f$. Then $h_{\psi^*}$ is a greatest element among secrecy-preserving joint information structures. 
\end{proposition}
\begin{figure}[ht]
\centering
\begin{tikzpicture}[x=1cm,y=1cm]
  \def\hh{0.123} 

  \node[anchor=east] at (-0.4, 1.0) {$\omega_1$};
  \node[anchor=east] at (-0.4, 0.0) {$\omega_2$};

  \draw[line width=0.8pt] (0,1.5) -- (7.0,1.5);

  \draw[dashed, gray!70] (1.5, 1.5) -- (1.5, -0.5);
  \draw[dashed, gray!70] (3.0, 1.5) -- (3.0, -0.5);
  \draw[dashed, gray!70] (4.0, 1.5) -- (4.0, -0.5);
  \draw[dashed, gray!70] (5.5, 1.5) -- (5.5, -0.5);

  \node[anchor=south] at (0.75, 1.5) {\small $y_1$};
  \node[anchor=south] at (2.25, 1.5) {\small $y_2$};
  \node[anchor=south] at (3.5, 1.5) {\small $y_3$};
  \node[anchor=south] at (4.75, 1.5) {\small $y_4$};
  \node[anchor=south] at (6.25, 1.5) {\small $y_5$};

  \fill[orange!40!yellow] (0,1.0-\hh) rectangle (4.0,1.0+\hh);
  \fill[blue!30!white]    (4.0,1.0-\hh) rectangle (5.5,1.0+\hh);
  \fill[red!35!white]     (5.5,1.0-\hh) rectangle (7.0,1.0+\hh);
  \draw[line width=2pt, black] (4.0,1.0-\hh) -- (4.0,1.0+\hh);
  \draw[line width=2pt, black] (5.5,1.0-\hh) -- (5.5,1.0+\hh);
  \node at (2.0, 1.0) {\footnotesize $x^*$};
  \node at (4.75, 1.0) {\footnotesize $x$};
  \node at (6.25, 1.0) {\footnotesize $x'$};

  \fill[red!35!white]     (0,0.0-\hh) rectangle (1.5,0.0+\hh);
  \fill[blue!30!white]    (1.5,0.0-\hh) rectangle (3.0,0.0+\hh);
  \fill[orange!40!yellow] (3.0,0.0-\hh) rectangle (7.0,0.0+\hh);
  \draw[line width=2pt, black] (1.5,0.0-\hh) -- (1.5,0.0+\hh);
  \draw[line width=2pt, black] (3.0,0.0-\hh) -- (3.0,0.0+\hh);
  \node at (0.75, 0.0) {\footnotesize $x'$};
  \node at (2.25, 0.0) {\footnotesize $x$};
  \node at (5.0, 0.0) {\footnotesize $x^*$};
\end{tikzpicture}
\caption{Greatest element among secrecy-preserving structures when $|\Omega|=2$}
\label{fig:secrecy-greatest}
\end{figure}
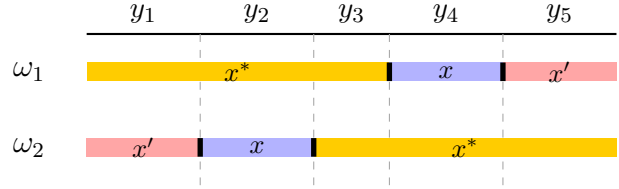
Figure \ref{fig:secrecy-greatest} illustrates the construction in Proposition \ref{prop:secrecy_greatest}. The argument is straightforward. By Proposition \ref{prop:secrecy-frontier}, it is enough to consider signal-based structures. When $|\Omega|=2$, at most one message can have total mass exceeding one across the two states. This message is $x^*$, if it exists. For every other message $x$, the total length required to place the $x$-regions across the two states is at most one, so these regions can be arranged without overlap. Any overlap for such a message is therefore avoidable and can be removed in a way that preserves the marginal distribution of each message while making the joint signal more informative. Hence, in a greatest secrecy-preserving structure, overlap can occur only for $x^*$ and the minimal unavoidable overlap for $x^*$ is
\begin{align*}
   \max\{f(x^*|\omega_1)+f(x^*|\omega_2)-1,0\}. 
\end{align*}
The representation $\psi^*$ attains exactly this overlap by placing the $x^*$-region for $\omega_1$ at one end of $T$ and the $x^*$-region for $\omega_2$ at the other, while arranging all other messages without overlap. Thus $h_{\psi^*}$ reveals the state whenever this is feasible and pools the two states only on the unavoidable overlap. It follows that $h_{\psi^*}$ Blackwell dominates every secrecy-preserving joint information structure, and hence is a greatest element.

This clean structure, however, does not extend to three or more states. In this case, even when $\sum_{\omega\in\Omega} f(x|\omega)\leq 1$ for some message $x$, a frontier signal-based structure may still assign overlapping $x$-regions across states. The following example illustrates this point.

Consider $\Omega=\{\omega_1,\omega_2,\omega_3\}$ and $X=\{x_1,x_2,x_3\}$, with baseline information structure
\[
f=
\begin{bmatrix}
1/3 & 2/3 & 0 \\
1/3 & 2/3 & 0 \\
1/3 & 1/3 & 1/3
\end{bmatrix}.
\]
Notice that $\sum_{\omega\in\Omega} f(x_1|\omega)=1$, so the $x_1$-regions could in principle be arranged without overlap across states. Partition $T=[0,1]$ into three intervals,
\[
Y_1=[0,1/3),\qquad
Y_2=[1/3,2/3),\qquad
Y_3=[2/3,1],
\]
and consider the signal representation $\psi:\Omega\times T\to X$ described by
\[
\begin{array}{c|ccc}
& Y_1 & Y_2 & Y_3\\
\hline
\omega_1 & x_1 & x_2 & x_2\\
\omega_2 & x_1 & x_2 & x_2\\
\omega_3 & x_2 & x_1 & x_3
\end{array}.
\]
The induced signal-based structure $h_\psi$ satisfies secrecy. Moreover, it lies on the Blackwell frontier. To see this, observe that $\omega_3$ is perfectly revealed. The information loss comes from the pooling of $\omega_1$ and $\omega_2$ on every interval. Any attempt to distinguish $\omega_1$ from $\omega_2$ requires rearranging their $x_2$-regions. But such a rearrangement necessarily creates new overlap with the $x_2$-region of $\omega_3$, thereby sacrificing information about $\omega_3$. Thus, no signal-based secrecy-preserving structure can strictly Blackwell dominate $h_\psi$.

Nevertheless, the $x_1$-regions overlap: both $\omega_1$ and $\omega_2$ assign $x_1$ on $Y_1$, even though $\sum_{\omega} f(x_1|\omega)=1$ would permit a non-overlapping assignment of $x_1$ across states. This shows that, with three or more states, the condition $\sum_{\omega} f(x|\omega)\leq 1$ does not imply that every frontier structure assigns the corresponding $x$-regions without overlap.

\bibliographystyle{econ}
\bibliography{references.bib}
\end{document}